\documentclass[a4paper]{article}

\usepackage[utf8]{inputenc}

\usepackage{fullpage}

\usepackage{microtype}

\usepackage[usenames,dvipsnames]{xcolor}

\usepackage[hidelinks,colorlinks,allcolors=blue]{hyperref}

\usepackage{todonotes}

\usepackage{amssymb,amsmath,amsthm,stmaryrd,amstext,wasysym}

\usepackage{tabularx}

\usepackage{caption,subcaption,graphicx,enumerate}

\usepackage{comment}

\usepackage{tikz} 
\usetikzlibrary{matrix}
\usetikzlibrary{arrows,shapes,positioning}
\usetikzlibrary{decorations.pathreplacing}

\def\cqedsymbol{\ifmmode$\lrcorner$\else{\unskip\nobreak\hfil
\penalty50\hskip1em\null\nobreak\hfil$\lrcorner$
\parfillskip=0pt\finalhyphendemerits=0\endgraf}\fi} 

\newcommand{\cqed}{\renewcommand{\qed}{\cqedsymbol}}

\newcommand{\pname}[1]{\textsc{#1}}
\newcommand{\cclass}[1]{\textsf{#1}}

\newcommand{\NP}{\cclass{NP}}
\newcommand{\TPE}{\pname{Trivially Perfect Editing}}
\newcommand{\TPD}{\pname{Trivially Perfect Deletion}}
\newcommand{\TPC}{\pname{Trivially Perfect Completion}}

\newcommand{\tpe}{\TPE}
\newcommand{\tpd}{\TPD}
\newcommand{\tpc}{\TPC}

\newcommand{\modulator}{TP-modulator}

\newcommand{\noinstance}{no-instance}

\renewcommand{\phi}{\varphi}

\newcommand{\third}{\gemini}  
\newcommand{\paw}{\mathsf{P}} 

\newcommand{\tpss}{TP-set system}

\DeclareMathOperator{\uni}{uni}   
\DeclareMathOperator{\poly}{poly} 

\newtheorem{theorem}{Theorem}
\newtheorem{lemma}{Lemma}[section]
\newtheorem{claim}[lemma]{Claim}

\newtheorem{corollary}[lemma]{Corollary}
\newtheorem{proposition}[lemma]{Proposition}

\theoremstyle{definition}
\newtheorem{definition}[lemma]{Definition}

\newtheorem{redrule}{Rule}

\newtheorem{innercustomrdrl}{Rule}
\newenvironment{customrdrl}[1]
{\renewcommand\theinnercustomrdrl{#1}\innercustomrdrl}
{\endinnercustomrdrl}

\theoremstyle{remark}
\newtheorem{observation}[lemma]{Observation}

\newtheorem*{hypothesis}{\textnormal{\textbf{Exponential Time
      Hypothesis}}}

\newcommand{\defparproblem}[4]{
  \hfill\\\smallskip\noindent%
  \begin{tabularx}{\textwidth}{|l X|}%
    \hline%
    \multicolumn{2}{|l|}{\pname{#1}}\\%
    \textbf{Input:}&#2\\%
    \textbf{Parameter:}&#3\\%
    \textbf{Question:}&#4\\\hline%
  \end{tabularx}%
  \smallskip%
}%

\title{A Polynomial Kernel for Trivially Perfect Editing\thanks{The research leading to these results has received funding from the European Research Council under the European Union's Seventh Framework Programme (FP/2007-2013) / ERC Grant Agreement n.~267959. M.~Pilipczuk is currently holding a post-doc position at Warsaw Center of Mathematics and Computer Science, and his research is supported by Polish National Science Centre grant DEC-2013/11/D/ST6/03073; However, large part of this work was done when M.~Pilipczuk was affiliated with the University of Bergen and was supported by the aforementioned ERC grant n.~267959.}}

\author{Pål Grønås Drange\thanks{
    Department~of Informatics, University of Bergen, Norway, \texttt{pal.drange@ii.uib.no}.
  } \and Michał Pilipczuk\thanks{
    Institute of Informatics, University of Warsaw, Poland, \texttt{michal.pilipczuk@mimuw.edu.pl}.
  }
}

\date{\today{}}

\begin{document}
\maketitle

\begin{abstract}
  We give a kernel with $O(k^7)$ vertices for \tpe, the problem of
  adding or removing at most $k$ edges in order to make a given graph
  trivially perfect.  This answers in affirmative an open question
  posed by Nastos and Gao~\cite{nastos2013familial}, and by Liu et
  al.~\cite{liu2014overview}.  Our general technique implies also the
  existence of kernels of the same size for related \tpc{} and \tpd{}
  problems.  Whereas for the former an $O(k^3)$ kernel was given by
  Guo~\cite{guo2007problem}, for the latter no polynomial kernel was
  known.

  We complement our study of \tpe{} by proving that, contrary to \tpc,
  it cannot be solved in time $2^{o(k)}\cdot n^{O(1)}$ unless the
  Exponential Time Hypothesis fails.  In this manner we complete the
  picture of the parameterized and kernelization complexity of the
  classic edge modification problems for the class of trivially
  perfect graphs.
\end{abstract}

\section{Introduction}
Graph modification problems form an important subclass of discrete
computational problems, where the task is to modify a given graph
using a constrained number of modifications in order to make it
satisfy some property~$\Pi$, or equivalently belong to the
class~$\mathcal{G}$ of graphs satisfying~$\Pi$.  Well-known examples
of graph modification problems include \pname{Vertex Cover},
\pname{Cluster Editing}, \pname{Feedback Vertex Set}, \pname{Odd Cycle
  Transversal}, and \pname{Minimum Fill-In}.  The systematic study of
graph modification problems dates back to early 80s and the work of
Yannakakis~\cite{Yannakakis81Edge}, who showed that there is a
dichotomy for the vertex deletion problems: unless a graph
class~$\mathcal{G}$ is trivial (finite or co-finite), the problem of
deleting the least number of vertices to obtain a graph
from~$\mathcal{G}$ is \NP-hard.  However, when, in order to obtain a
graph from~$\mathcal{G}$, we are to modify the edge set of the graph
instead of the vertex set, there are three natural classes of
problems: deletion problems (deleting the least number of edges),
completion problems (adding the least number of edges) and editing
problems (performing the least number of edge additions or deletions).
For neither of these, any complexity dichotomy in the spirit of
Yannakakis' result is known.  Indeed, in~\cite{Yannakakis81Edge}
Yannakakis states
\begin{quote}
  \textit{It [\dots] would be nice if the same kind of techniques could be
    applied to the edge-deletion problems.  Unfortunately we suspect
    that this is not the case --- the reductions we found for the
    properties considered [\dots] do not seem to fall into a pattern.}
  \\--- Mihalis Yannakakis
\end{quote}

Even though for edge modification problems there is no general
\cclass{P}~vs.~\NP{} classification known, much can be said about
their parameterized complexity.  Recall that a parameterized problem
is called \emph{fixed-parameter tractable} if it can be solved in time
$f(k)\cdot n^{O(1)}$ for some computable function~$f$, where~$n$ is
the size of the input and~$k$ is its parameter.  In our case, the
natural parameter~$k$ is the allowed number of modifications.
Cai~\cite{cai1996fixed} made a simple observation that for all the
aforementioned graph modification problems there is a simple branching
algorithm running in time $c^k n^{O(1)}$ for some constant~$c$, as
long as $\mathcal{G}$ is \emph{characterized by a finite set of
  forbidden induced subgraphs}: there is a finite list of graphs
$H_1,H_2,\ldots,H_p$ such that any graph~$G$ belongs to~$\mathcal{G}$
if and only if~$G$ does not contain any~$H_i$ as an induced subgraph.
Although many studied graph classes satisfy this property, there are
important examples, like chordal or interval graphs, that are outside
this regime.

For this reason, the parameterized analysis of modification problems
for graph classes characterized by a finite set of forbidden induced
subgraphs focused on studying the design of \emph{polynomial
  kernelization algorithms} (\emph{polynomial kernels}); Recall that
such an algorithm is required, given an input instance~$(G,k)$ of the
problem, to preprocess it in polynomial time and obtain an equivalent
output instance~$(G',k')$, where $|G'|,k' \leq p(k)$ for some
polynomial~$p$.  That is, the question is the following: can you,
using polynomial-time preprocessing only, bound the size of the
tackled instance by a polynomial function depending only on~$k$?

For vertex deletion problems the answer is again quite simple: As long
as~$\mathcal{G}$ is characterized by a finite set of forbidden induced
subgraphs, the task is to hit all the copies of these subgraphs
(so-called \emph{obstacles}) that are originally contained in the
graph.  Hence, one can construct a simple reduction to the
\pname{$d$-Hitting Set} problem for a constant~$d$ depending
on~$\mathcal{G}$, and use the classic~$O(k^d)$ kernel for the latter
that is based on the sunflower lemma (see
e.g.~\cite{FlumGrohebook,FominSV13}).  For edge modifications
problems, however, this approach fails utterly: every edge
addition/deletion can create new obstacles, and thus it is not
sufficient to hit only the original ones.  For this reason, edge
modification problems behave counterintuitively w.r.t.\ polynomial
kernelization, and up to recently very little was known about their
complexity.

On the positive side, kernelization of edge modification problems for
well-studied graph classes was explored by Guo~\cite{guo2007problem},
who showed that four problems: \pname{Threshold Completion},
\pname{Split Completion}, \pname{Chain Completion}, and
\pname{Trivially Perfect Completion}, all admit polynomial kernels.
However, the study took a turn for the interesting when Kratch and
Wahlström~\cite{kratsch2009two} showed that there is a graph~$H$ on
$7$ vertices, such that the deletion problem to $H$-free graphs (the
class of graphs not admitting~$H$ as an induced subgraph) does not
admit a polynomial kernel, unless the polynomial hierarchy collapses.
This shows that the subtle differences between edge modification and
vertex deletion problems have tremendous impact on the kernelization
complexity.

Kratch and Wahlström conclude by asking whether there is a ``simple''
graph, like a path or a cycle, for which an edge modification problem
does not admit a polynomial kernel under similar assumptions.  The
question was answered by Guillemot et al.~\cite{guillemot2013onthenon}
who showed that both for the class of $P_\ell$-free graphs (for~$\ell
\geq 7$) and for the class of~$C_\ell$-free graphs (for~$\ell\geq 4$),
the edge deletion problems probably do not have polynomial
kernelization algorithms.  They simultaneously gave a cubic kernel for
the \pname{Cograph Editing} problem, the problem of editing to a graph
without induced paths on four vertices.

These results were later improved by Cai and
Cai~\cite{cai2013incompressibility}, who tried to obtain a complete
dichotomy of the kernelization complexity of edge modification
problems for classes of~$H$-free graphs, for every graph~$H$.  The
project has been almost fully successful --- the question remains
unresolved only for a finite number of graphs~$H$.  In particular, it
turns out that the existence of a polynomial kernel for any of
\pname{$H$-Free Editing}, \pname{$H$-Free Edge Deletion}, or
\pname{$H$-Free Completion} problems is in fact a very rare
phenomenon, and basically happens only for specific, constant-size
graphs~$H$.  In particular, for~$H$ being a path or a cycle, the
aforementioned three problems admit polynomial kernels if and only
if~$H$ has at most three edges.

\bigskip

At the same time, there is a growing interest in identifying
parameterized problems that are solvable in \emph{subexponential
  parameterized time}, i.e., in time $2^{o(k)}n^{O(1)}$.  Although for
many classic parameterized problems already known \NP-hardness
reductions show that the existence of such an algorithm would
contradict the \emph{Exponential Time Hypothesis} of Impagliazzo et
al.~\cite{impagliazzo2001which}, subexponential parameterized
algorithms were known to exist for problems in restricted settings,
like planar, or more generally $H$-minor free
graphs~\cite{demaine2005subexponential}, or
tournaments~\cite{alon2009fast}.  See the book of Flum and
Grohe~\cite{FlumGrohebook} for a wider discussion.

Therefore, it was an immense surprise when Fomin and
Villanger~\cite{fomin2012subexponential} showed that \pname{Chordal
  Completion} (also called \pname{Minimum Fill-In}) can be solved in
time $2^{O(\sqrt{k}\log k)}n^{O(1)}$.  Following this discovery, a new
line of research was initiated.  Ghosh et al.~\cite{ghosh2013faster}
showed that \pname{Split Completion} is solvable in the same running
time.  Although Komusiewicz and
Uhlmann~\cite{komusiewicz2012clusterediting} showed that we cannot
expect \pname{Cluster Editing} to be solvable in subexponential
parameterized time, as shown by Fomin et
al.~\cite{fomin2011subexponential}, when the number of clusters in the
target graph is sublinear in the number of allowed edits, this is
possible nonetheless.

Following these three positive examples, Drange et
al.~\cite{drange2014exploring} showed that completion problems for
trivially perfect graphs, threshold graphs and pseudosplit graphs all
admit subexponential parameterized algorithms.  Later, Bliznets et al.
showed that both \pname{Proper Interval Completion} and
\pname{Interval Completion} also admit subexponential parameterized
algorithms~\cite{bliznets2014proper,bliznets2014interval}.

Let us remark that in almost all these results, the known existence of
a polynomial kernelization procedure for the problem played a vital
role in designing the subexponential parameterized algorithm.
Kernelization is namely used as an opening step that enables us to
assume that the size of the considered graph is polynomial in the
parameter $k$, something that turns out to be extremely useful in
further reasonings.  The only exception is the algorithm for the
\pname{Interval Completion} problem~\cite{bliznets2014interval}, for
which the existence of a polynomial kernel remains a notorious open
problem.  The need of circumventing this issue created severe
difficulties in~\cite{bliznets2014interval}.

\bigskip

In this paper we study the \pname{Trivially Perfect Editing} problem.
Recall that trivially perfect graphs are exactly graphs that do not
contain a~$P_4$ or a~$C_4$ as an induced subgraph; see
Section~\ref{ssec:prelim-tp} for a structural characterization of this
graph class.  Interest in trivially perfect graphs started with the
attempts to prove the strong perfect graph theorem.  In recent times,
new source of motivation has grown, with the realization that
trivially perfect graphs are related to the width parameter
\emph{treedepth} (called also vertex ranking number, ordered chromatic
number, and minimum elimination tree height).  Although it had been
known that both the completion and the deletion problem for trivially
perfect graphs are \NP-hard, it was open for a long time whether the
editing version is \NP-hard as
well~\cite{burzyn06NPcompleteness,mancini2008graph}.

This question was answered very recently by Nastos and
Gao~\cite{nastos2013familial}, who showed that the problem is indeed
\NP-hard.  Actually, the work of Nastos and Gao focuses on exhibiting
applications of trivially perfect graphs in social network theory,
since this graph class may serve as a model for \emph{familial
  groups}, communities in social networks showing a hierarchical
nature.  Specifically, the \emph{editing number} to a trivially
perfect graph\footnote{Nastos and Gao use the terminology
  \emph{quasi-threshold graphs} instead of trivially perfect graphs.}
can be used as a measure of how much a social network resembles a
collection of hierarchies.  Nastos and Gao also ask whether it is
possibly to obtain a polynomial kernelization algorithm for this
problem.  The question about the existence of a polynomial kernel for
\TPE{} was then restated in a recent survey by Liu, Wang, and
Guo~\cite{liu2014overview}, which \textit{nota bene} contains a
comprehensive overview of the current status of the research on the
kernelization complexity of graph modification problems.

%
%
\paragraph{Our contribution.}
We answer the question of Nastos and Gao~\cite{nastos2013familial} and
of Liu, Wang, and Guo~\cite{liu2014overview} in affirmative by proving
the following theorem.

\begin{theorem}
  \label{thm:tpe-polykernel-intro}
  The problem \TPE{} admits a proper kernel with~$O(k^7)$ vertices.
\end{theorem}

Here, we say that a kernel (kernelization algorithm) is {\em{proper}}
if it can only decrease the parameter, i.e., the output parameter~$k'$
satisfies~$k'\leq k$.

To prove Theorem~\ref{thm:tpe-polykernel-intro}, we employ an
extensive analysis of the tackled instance, based on the equivalent
structural definition of trivially perfect graphs.  The main approach
is to construct a small \emph{vertex modulator}, a set of vertices
whose removal results in obtaining a trivially perfect graph.
However, since we are allowed only edge deletions and additions, this
modulator just serves as a tool for exposing the structure of the
instance.  More specifically, we greedily pack disjoint obstructions
into a set~$X$, whose size can be guaranteed to be at most~$4k$, with
the condition that to get rid of each of these obstructions, at least
one edge must be edited inside the modulator per obstruction.  Having
obtained such a modulator, the rest of the graph,~$G-X$, is trivially
perfect, and we may apply the structural view on trivially perfect
graphs to find irrelevant parts that can be reduced.

While the modulator technique is commonly used in kernelization, the
new insight in this work is as follows.  Since we work with an edge
modification problem, we can be less restrictive about when an
obstacle can be greedily packed into the modulator.  For example, the
obstacle does not need to be completely vertex-disjoint with the so
far constructed~$X$; sharing just one vertex is still allowed.  This
observation allows us to reason about the adjacency structure
between~$X$ and~$V(G)\setminus X$, which is of great help when
identifying irrelevant parts.  We hope that this generic technique
finds applications in other edge modification problems as well.

By slight modifications of our kernelization algorithm, we also obtain
polynomial kernels for \TPD{} and \TPC{}.

\begin{theorem}
  \label{thm:tpd-polykernel-intro}
  The problem \TPD{} admits a proper kernel with~$O(k^7)$ vertices.
\end{theorem}
\begin{theorem}
  \label{thm:tpc-polykernel-intro}
  The problem \TPC{} admits a proper kernel with~$O(k^7)$ vertices.
\end{theorem}

To the best of our knowledge, no polynomial kernel for \TPD{} was
known so far.  For \TPC, a cubic kernel was shown earlier by
Guo~\cite{guo2007problem}.  Unfortunately, the work of
Guo~\cite{guo2007problem} is published only as a conference extended
abstract, where it is only sketched how the approach yielding a
quartic kernel for \pname{Split Deletion} could be used to obtain a
cubic kernel for \TPC.  The details of this kernelization algorithm
are deferred to the full version, which, alas, has not appeared.  For
this reason, we believe that our proof of
Theorem~\ref{thm:tpc-polykernel-intro} fills an important gap in the
literature --- the polynomial kernel for \TPC{} is an important
ingredient of the subexponential parameterized algorithm for this
problem~\cite{drange2014exploring}.

Finally, we show that \TPE{}, in addition to being \NP-complete,
cannot admit a subexponential parameterized algorithm, provided that
the Exponential Time Hypothesis holds.
\begin{theorem}
  \label{thm:eth-hardness-intro}
  Under ETH, the \tpe{} problem is \cclass{NP}-hard and cannot be
  solved in time $2^{o(k)}n^{O(1)}$ or $2^{o(n+m)}$ even on graphs
  with maximum degree~$4$.
\end{theorem}
In other words; the familial group measure cannot be computed in time
subexponential in terms of the value of the measure.  This stands in
contrast with \tpc{} that admits a subexponential parameterized
algorithm, and shows that \tpe{} is more similar to \tpd, for which a
similar lower bound has been proved earlier by Drange et
al.~\cite{drange2014exploring}.  In fact, our reduction can be used as
an alternative proof of hardness of \tpd{} as well.

Let us note that the \NP-hardness reduction for \tpe{} presented by
Nastos and Gao~\cite{nastos2013familial} cannot be used to prove
nonexistence of a subexponential parameterized algorithm, since it
involves a cubic blow-up of the parameter (see
Section~\ref{sec:hardness} for details).  To prove
Theorem~\ref{thm:eth-hardness-intro}, we resort to the technique used
for similar hardness results by Komusiewicz and
Uhlmann~\cite{komusiewicz2012clusterediting} and by Drange et
al.~\cite{drange2014exploring}.

\section{Preliminaries}
\label{sec:prelim}
\subsection{Graphs and complexity}
\label{ssec:prelim-graphs-complexity}
\paragraph{Graphs.}
In this work we consider only undirected simple finite graphs.  For a
graph $G$, by $V(G)$ and $E(G)$ we denote the vertex and edge set of
$G$, respectively. The {\em{size}} of a graph $G$ is defined as $|G|=|V(G)|+|E(G)|$. 

For a vertex $v \in V(G)$, by $N_G(v)$ we denote
the open neighborhood of~$v$, i.e.  $N_G(v)=\{u \in V(G) \mid uv \in
E(G)\}$.  The closed neighborhood of $v$, denoted by $N_G[v]$, is
defined as $N_G(v)\cup \{v\}$.  These notions are extended to subsets
of vertices as follows: $N_G[X]=\bigcup_{v\in X} N_G[v]$ and
$N_G(X)=N_G[X]\setminus X$.  We omit the subscript whenever~$G$ is
clear from context.

When $U\subseteq V(G)$ is a subset of vertices of~$G$, we write~$G[U]$ to
denote the \emph{induced subgraph} of~$G$, i.e., the graph $G' =
(U,E_U)$ where~$E_U$ is~$E(G)$ restricted to~$U$.  The degree of a
vertex $v \in V(G)$, denoted $\deg_G(v)$, is the number of vertices it
is adjacent to, i.e., $\deg_G(v) = |N_G(v)|$.  We denote by
$\Delta(G)$ the maximum degree in the graph, i.e., $\Delta(G) =
\max_{v \in V(G)}\deg(v)$.  For a set $A$, we write $\binom{A}{2}$ to
denote the set of unordered pairs of elements of $A$; thus $E(G)
\subseteq \binom{V(G)}{2}$.  By $\overline{G}$ we denote the
\emph{complement} of a graph~$G$, i.e., $V(\overline{G})=V(G)$ and
$E(\overline{G})=\binom{V(G)}{2}\setminus E(G)$.

If~$v$ and~$u$ are such that $N[v] = N[u]$, then we call~$v$ and~$u$
\emph{true twins}.  Observe that~$v$ and~$u$ are adjacent if they are
true twins.  On the other hand, if~$v$ and~$u$ have~$N(v) = N(u)$,
then we call~$v$ and~$u$ \emph{false twins}, and in this case we may
observe that~$v$ and~$u$ are non-adjacent.  If~$X$ is an
inclusion-wise maximal set of vertices such that for every pair of
vertices~$v$ and~$u$ in~$X$ they are true (resp.\ false) twins, then
we call~$X$ a true (resp.\ false) twin class.

For a graph $G$ and a set of vertices $X \subseteq V(G)$, we denote by $G-X$
the (induced subgraph) $G[V(G) \setminus X]$.  When $F \subseteq
\binom{V(G)}{2}$, we write~$G-F$ to denote the graph~$G'$ on vertex
set~$V(G)$ and edge set $E(G) \setminus F$.  Finally, we let $G
\triangle F$ be the graph on vertex set~$V(G)$ and edge set $E(G)
\triangle F$, where $\triangle$ denotes the \emph{symmetric
  difference}; For two sets~$A$ and~$B$, $A \triangle B = (A \setminus
B) \cup (B \setminus A)$.  We will also say that two sets $A$ and $B$
are \emph{nested} if $A \subseteq B$ or $B \subseteq A$.

A vertex $v\in V(G)$ is \emph{universal} if it is adjacent to all the
other vertices of the graph.  Note that the set of universal vertices
of a graph forms a clique, which is also a true twin class.  This
clique will be denoted by $\uni(G)$ and called the \emph{universal
  clique} of $G$.

\paragraph{Modules and the modular decomposition.}

In our kernelization algorithm we will use the notion of a
\emph{module} in a graph.

\begin{definition}
  \label{def:module}
  Given a graph~$G$, a set of vertices~$M \subseteq V(G)$ is called a
  \emph{module} if for any two vertices~$v$ and~$u$ in~$M$, we have
  that $N(v) \setminus M = N(u) \setminus M$, i.e., all the vertices
  of $M$ have exactly the same neighborhood outside $M$.
\end{definition}

Observe that for any graph $G$, any singleton~$M=\{v\}$ is a module,
and also~$V(G)$ itself is a module.  However, $G$ can contain a whole
hierarchy of modules.  This hierarchy can be captured using the
following notion of a \emph{modular decomposition}, introduced by
Gallai~\cite{modular-decomp}.  The following description of a modular
decomposition is taken verbatim from the work of Bliznets et
al.~\cite{bliznets2014interval} .

A module decomposition of a graph $G$ is a rooted tree $T$, where each
node $t$ is labeled by a module $M^t \subseteq V(G)$, and is one of
four types:
\begin{description}
\item[leaf:] $t$ is a leaf of $T$, and $M^t$ is a singleton;
\item[union:] $G[M^t]$ is disconnected, and the children of $t$ are
  labeled with different connected components of $G[M^t]$;
\item[join:] the complement of $G[M^t]$ is disconnected, and the
  children of $t$ are labeled with different connected components of
  the complement of $G[M^t]$;
\item[prime:] neither of the above holds, and the children of $t$ are
  labeled with different modules of $G$ that are proper subsets of
  $M^t$, and are inclusion-wise maximal with this property.
\end{description}
Moreover, we require that the root of $T$ is labeled with the module
$V(G)$.  We need the following properties of the module decomposition.
\begin{theorem}[see
  \cite{compute-modular-decomp}]\label{thm:module-decomp}
  For a graph $G$, the following holds.
  \begin{enumerate}
  \item A module decomposition $(T,(M^t)_{t \in V(T)})$ of $G$ exists,
    is unique, and computable in linear time.
  \item At any prime node $t$ of $T$, the labels of the children form
    a partition of $M^t$.  In particular, for each vertex $v$ of $G$
    there exists exactly one leaf node with label $\{v\}$.
  \item Each module $M$ of $G$ is either a label of some node of $T$,
    or there exists a \textbf{union} or \textbf{join} node $t$ such
    that $M$ is a union of labels of some children of $t$.
  \end{enumerate}
\end{theorem}

Let us remark that since in this work we do not optimize the running
time of the kernelization algorithm, we do not need to compute the
modular decomposition in linear time.  Any simpler polynomial time
algorithm would suffice (see the work of McConnell and
Spinrad~\cite{compute-modular-decomp} for a literature overview).


\paragraph{Parameterized complexity.}
The running time of an algorithm is usually described as a function of
the length of the input.  To refine the complexity analysis of
computationally hard problems, parameterized complexity introduced the
notion of an extra ``parameter'' that is an additional part of a
problem instance responsible for measuring its complexity.  To
simplify the notation, we will consider inputs to problems of the form
$(G,k)$, which is a pair consisting of a graph~$G$ and a nonnegative
integer~$k$.  A problem is then said to be \emph{fixed parameter
  tractable} if there is an algorithm which solves the problem in time
$f(k) \cdot \poly(|G|)$, where $f$ is any function, and $\poly \colon
\mathbb{N} \to \mathbb{N}$ any polynomial function.  In the case when
$f(k) = 2^{o(k)}$ we say that the algorithm is a subexponential
parameterized algorithm.  When a problem $\Pi \subseteq \mathcal{G}
\times \mathbb{N}$ is fixed-parameter tractable, where~$\mathcal{G}$
is the class of all graphs, we say that~$\Pi$ belongs to the
complexity class \cclass{FPT}.  For a more rigorous introduction to
parameterized complexity we refer to the books of Downey and
Fellows~\cite{DowneyF99} and of Flum and Grohe~\cite{FlumGrohebook}.

A \emph{kernelization algorithm} (or \emph{kernel}) is a
polynomial-time algorithm for a parameterized problem~$\Pi$ that takes
as input a problem instance $(G,k)$ and returns an equivalent instance
$(G',k')$, i.e.  $(G,k)\in \Pi \Leftrightarrow (G',k')\in \Pi$, where
both~$|G'|$ and~$k'$ are bounded by~$f(k)$ for some function~$f$.  We
then say that~$f$ is the \emph{size of the kernel}.  When $k' \leq k$,
we say that the kernel is a \emph{proper kernel}.  Specifically, a
proper polynomial kernelization algorithm for~$\Pi$ is a polynomial
time algorithm which takes as input an instance $(G,k)$ and returns an
equivalent instance $(G',k')$ with $k' \leq k$ and $|G'| \leq p(k)$
for some polynomial function $p$.

\paragraph{Tools for lower bounds.}
As evidence that \pname{Trivially Perfect Editing} cannot be solved in
subexponential parameterized time $2^{o(k)} n^{O(1)}$, we will use the
Exponential Time Hypothesis, formulated by Impagliazzo, Paturi, and
Zane~\cite{impagliazzo2001which}.

\begin{hypothesis}[Exponential Time Hypothesis, ETH] There exists a
  positive real number $s$ such that \pname{3Sat} with $n$ variables
  and $m$ clauses cannot be solved in time $2^{sn}(n + m)^{O(1)}$.
\end{hypothesis}

Impagliazzo, Paturi, and Zane~\cite{impagliazzo2001which} proved a
fundamental result called \emph{Sparsification Lemma}, which can serve
as a Turing reduction from an arbitrary instance of \pname{3Sat} to an
instance where the number of clauses is linear in the number of
variables.  Thus, the following statement is an immediate corollary of
the Sparsification Lemma.

\begin{proposition}[\cite{impagliazzo2001which}]\label{prop:eth}
  Unless ETH fails, there exists a positive real number $s$ such that
  \pname{3Sat} with $n$ variables and $m$ clauses cannot be solved in
  time $2^{s(n+m)}(n + m)^{O(1)}$.  In particular, \pname{3Sat} does
  not admit an algorithm with time complexity
  $2^{o(n+m)}(n+m)^{O(1)}$.
\end{proposition}

\subsection{Trivially Perfect Graphs}
\label{ssec:prelim-tp}

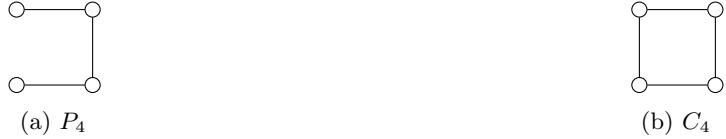
\begin{figure}[t]
  \centering
  \begin{subfigure}[t]{0.45\textwidth}
    \centering
    \begin{tikzpicture}[every node/.style={circle, draw, scale=.6},
      scale=1]
      \node (1) at (0,0) {};
      \node (2) at (1,0) {};
      \node (3) at (0,1) {};
      \node (4) at (1,1) {};
      
      \draw (1) -- (2);
      \draw (2) -- (4);
      \draw (3) -- (4);
    \end{tikzpicture}
    \caption{$P_4$}
  \end{subfigure}
  \hspace{.05\textwidth}
  \begin{subfigure}[t]{0.45\textwidth}
    \centering
    \begin{tikzpicture}[every node/.style={circle, draw, scale=.6},
      scale=1]
      \node (1) at (0,0) {};
      \node (2) at (1,0) {};
      \node (3) at (0,1) {};
      \node (4) at (1,1) {};
      
      \draw (1) -- (2);
      \draw (2) -- (4);
      \draw (3) -- (4);
      \draw (1) -- (3);
    \end{tikzpicture}
    \caption{$C_4$}
  \end{subfigure}
  \caption{\emph{Trivially perfect graphs} are $\{C_4, P_4\}$-free.}
  \label{fig:forbidden-graphs}
\end{figure}

\paragraph{Combinatorial properties.} A graph~$G$ is trivially perfect
if and only if it does not contain a~$C_4$ or a~$P_4$ as an induced
subgraph.  That is, trivially perfect graphs are defined by the
forbidden induced subgraph family $F = \{C_4,P_4\}$ (see
Figure~\ref{fig:forbidden-graphs}).  However, we mostly rely on the
following recursive characterization of the trivially perfect graphs:

\begin{proposition}[\cite{jing1996quasi}]
  \label{def:tp-recursive}
  The class of trivially perfect graphs can be defined recursively as
  follows:
  \begin{itemize}
  \item $K_1$ is a trivially perfect graph.
  \item Adding a universal vertex to a trivially perfect graph results
    in a trivially perfect graph.
  \item The disjoint union of two trivially perfect graphs results in
    a trivially perfect graph.
  \end{itemize}
\end{proposition}

Based on Proposition~\ref{def:tp-recursive}, a superset of the current
authors~\cite{drange2014exploring} proposed the following notion of a
decomposition for trivially perfect graphs.  In the following, for a
rooted tree $T$ and vertex $t\in V(T)$, by $T_t$ we denote the subtree
of $T$ rooted at $t$.

\begin{definition}[Universal clique decomposition,
  \cite{drange2014exploring}]
  \label{def:univeral-clique-decomposition}
  A \emph{universal clique decomposition} (\emph{UCD}) of a connected
  graph $G$ is a pair $\mathcal{T} = (T=(V_T,E_T),
  \mathcal{B}=\{B_{t}\}_{t\in V_T})$, where~$T$ is a rooted tree
  and~$\mathcal{B}$ is a partition of the vertex set~$V(G)$ into
  disjoint nonempty subsets, such that
  \begin{itemize}
  \item if~$vw \in E(G)$ and~$v \in B_t,w \in B_s$, then either~$t = s$,~$t$
    is an ancestor of~$s$ in~$T$, or~$s$ is an ancestor of~$t$ in~$T$,
    and
  \item for every node~$t \in V_T$, the set of vertices~$B_t$ is the
    universal clique of the induced subgraph $G[\bigcup_{s\in V(T_t)}
    B_s]$.
  \end{itemize}
\end{definition}

We call the vertices of~$T$ \emph{nodes} and the sets in~$\mathcal{B}$
\emph{bags} of the universal clique decomposition~$(T, \mathcal{B})$.
By slightly abusing notation, we often identify nodes with
corresponding bags.  Note that by the definition, in a universal
clique decomposition every non-leaf node~$t$ has at least two
children, since otherwise the bag~$B_t$ would not comprise \emph{all}
the universal vertices of the graph~$G[\bigcup_{s\in V(T_t)} B_s]$.

The following lemma explains the connection between trivially perfect
graphs and universal clique decompositions.

\begin{lemma}[\cite{drange2014exploring}]
  \label{lem:tp-iff-ucd}
  A connected graph~$G$ admits a universal clique decomposition if and
  only if it is trivially perfect.  Moreover, such a decomposition is
  unique up to isomorphisms.
\end{lemma}

Note that a universal clique decomposition can trivially be found in
polynomial time by repeatedly locating universal vertices and
connected components.  Moreover, we can extend the notion of a
universal clique decomposition also to a disconnected trivially
perfect graph $G$.  In this case, the universal clique decomposition
of $G$ becomes a rooted forest consisting of universal clique
decompositions of the connected components of $G$.  Since a graph is
trivially perfect if and only if each of its connected component is,
Lemma~\ref{lem:tp-iff-ucd} can be easily generalized to the following
statement: Every (possibly disconnected) graph $G$ is trivially
perfect if and only if it admits a universal clique decomposition,
where the decomposition has the shape of a rooted forest.  Moreover,
this decomposition is unique up to isomorphism.

The following definition of a quasi-ordering of vertices respecting
the UCD will be helpful when arguing the correctness of the
kernelization procedure.
\begin{definition}
  Let $(T,\mathcal{B})$ be the universal clique decomposition of a
  trivially perfect graph $G$.  We impose a quasi-ordering~$\preceq$
  on vertices of $G$ defined as follows.  Suppose vertex $u$ belongs
  to bag $B_t$ and vertex $v$ belongs to bag $B_s$.  Then $u\preceq v$
  if and only if $t=s$ or $t$ is an ancestor of $s$ in the rooted
  forest $T$.
\end{definition}

Thus, classes of vertices pairwise equivalent with respect
to~$\preceq$ are exactly formed by the bags of $\mathcal{B}$, and
otherwise the ordering respects the rooted structure of $T$.  Note
that since the UCD of a trivially perfect graph is unique up to
isomorphism, the quasi-ordering $\preceq$ is uniquely defined and can
be computed in polynomial time.

\paragraph{Computational problems.} In this work we are mainly
interested in the \tpe{} problem, defined formally as follows:

\defparproblem{Trivially Perfect Editing}
{A graph $G$ and a non-negative integer $k$.}
{$k$}
{Is there a set $S \subseteq \binom{V(G)}{2}$ of size at most $k$ such that $G
  \triangle S$ is trivially perfect?}

For a graph $G$, any set $F \subseteq \binom{V(G)}{2}$ for which $G \triangle
F$ is trivially perfect will henceforth be referred to as an
\emph{editing set}.  An editing set is \emph{minimal} if no proper
subset~$F' \subsetneq F$ is also an editing set.

In the \TPD{} and \TPC{} problems we allow only edge deletions and
edge additions, respectively.  More formally, we require that the
editing set $S$ is contained in, or disjoint from $E(G)$,
respectively.  In Section~\ref{sec:kernel-editing} we prove
Theorem~\ref{thm:tpe-polykernel-intro}, that is, we show that \tpe{}
admits a kernel with $O(k^7)$.
%
%
Actually, the character of our data reduction rules will be very
simple;
The kernelization algorithm will start with instance $(G,k)$, and
perform only the following operations:
\begin{itemize}
\item edit some $e\in \binom{V(G)}{2}$, decrement the budget $k$ by $1$, and
  terminate the algorithm if $k$ becomes negative; or
\item remove some vertex $u$ of $G$ and proceed with instance
  $(G-u,k)$.
\end{itemize}
Thus, the kernel will essentially be an induced subgraph of $G$,
modulo performing some edits whose safeness and necessity can be
deduced.  In the proofs of correctness, we will never use any
minimality argument that exchanges edge deletions for completions, or
vice versa.  Therefore, the whole approach can be applied almost
verbatim to \tpd{} and \tpc{}, yielding proofs for
Theorems~\ref{thm:tpd-polykernel-intro}
and~\ref{thm:tpc-polykernel-intro} after very minor modifications.  We
hope that the reader will be convinced about this after understanding
all the arguments of Section~\ref{sec:kernel-editing}.
%
%
However, for the sake of completeness we, in
Section~\ref{sec:kernel-comp-del}, review the modifications of the
argumentation of Section~\ref{sec:kernel-editing} that are necessary
to prove Theorems~\ref{thm:tpd-polykernel-intro}
and~\ref{thm:tpc-polykernel-intro}.
%

%
%
%



\paragraph{\tpss s.}
In the kernelization algorithm we will need the following auxiliary
definition and result.

\begin{definition}[\tpss]
  \label{def:tpss}
  A set system $\mathcal{F}\subseteq 2^U$ over a ground set $U$ is called a
  \emph{\tpss} if for every $X_1$ and $X_2$ in $\mathcal{F}$ with $x_1
  \in X_1 \setminus X_2$ and $x_2 \in X_2 \setminus X_1$, there is no
  $Y \in \mathcal{F}$ with $\{x_1, x_2\} \subseteq Y$.
\end{definition}

The following property bounds the size of a \tpss, which we need
later:

\begin{lemma}
  \label{lem:tpss-bounded}
  Let~$\mathcal{F}$ be a \tpss{} over a finite ground set~$U$.  Then
  the cardinality of~$\mathcal{F}$ is at most~$|U| + 1$.
\end{lemma}
\begin{proof}
  We proceed by induction on $|U|$, with the claim being trivial when
  $U=\emptyset$.
  Suppose $\mathcal{F}$ is a \tpss{} over a ground set $U$, and let
  $X$ be a member of $\mathcal{F}$ that has the minimum cardinality
  among the nonempty ones (if there is no such set, then
  $|\mathcal{F}|\leq 1$ and we are done).  The first observation is
  that if $Y_1$ and $Y_2$ are two nonempty members of $\mathcal{F}$
  that satisfy $Y_1 \setminus X = Y_2 \setminus X$ (possibly $Y_1=X$ or $Y_2=X$), then in fact $Y_1
  = Y_2$.  Suppose otherwise that there exist two such nonempty sets
  $Y_1,Y_2\in \mathcal{F}$ with $Y_1\cap X\neq Y_2\cap X$; W.l.o.g.,
  suppose that there exists an element $x_1\in Y_1\setminus Y_2
  \subseteq X$, and hence $x_1 \in X \setminus Y_2$.
  Since $X$ is of minimum cardinality, we have that $|X|\leq |Y_2|$.  As
  $X\nsubseteq Y_2$, we infer that there exists an element $x_2\in
  Y_2\setminus X=Y_1\setminus X$.  Consider the pair $\{x_1,x_2\}$ and
  observe that (a)~$x_1\in X\setminus Y_2$, (b)~$x_2 \in Y_2 \setminus
  X$, and (c)~$\{x_1,x_2\} \subseteq Y_1$.  This contradicts the
  definition of a \tpss.

  Define a set system $\mathcal{F}'$ over the ground set $U\setminus X$ as
  follows:
  $$\mathcal{F}'=\{Y \setminus X\ \colon\ Y \in \mathcal{F},\, Y\neq \emptyset\}.$$
  Clearly, $\mathcal{F}'$ is a \tpss{} over a strictly smaller ground
  set, so from the induction hypothesis we infer that
  $|\mathcal{F}'|\leq |U\setminus X|+1$.  Moreover, from the
  observation of the previous paragraph we infer that sets $Y\setminus
  X$ are pairwise different for $Y \in \mathcal{F}, Y\neq \emptyset$,
  and hence $|\mathcal{F}|\leq |\mathcal{F}'|+1$ (the additive $+1$
  comes from possibly having the empty set in $\mathcal{F}$).
  Concluding,
  \[
  |\mathcal{F}| \leq |\mathcal{F}'| + 1 \leq |U\setminus X|+1+1 \leq |U|-1+1+1=|U|+1.
  \]
\end{proof}

\section{A Kernel for \pname{Trivially Perfect Editing}}
\label{sec:kernel-editing}

This section is devoted to the proof of
Theorem~\ref{thm:tpe-polykernel-intro}, stating that \tpe{} admits a
proper kernel with $O(k^7)$ vertices.  As usual, the kernelization
algorithm will be given as a sequence of \emph{data reduction rules}:
simple preprocessing procedures that, if applicable, simplify the
instance at hand.  For each rule we shall prove two results: (a) that
applicability of the rule can be recognized in polynomial time, and
(b) that the rule is \emph{safe}, i.e., the resulting instance is
equivalent to the input one.  At the end of the proof we will argue
that if no rule is applicable, then the size of the instance must be
bounded by~$O(k^7)$.  Some rules will decrement the budget~$k$ for
edge edits; If this budget drops below zero, we may conclude that we
are dealing with a no-instance, so we immediately terminate the
algorithm and provide a constant-size trivial no-instance as the
obtained kernel, for example the instance $(C_4, 0)$.

Before starting the formal description, let us give a brief overview
of the structure of the proof.  In Section~\ref{sec:basic-rules} we
give some preliminary basic rules, which mostly deal with situations
where we can find a large number of induced~$C_4$s and~$P_4$s in the
graph (henceforth called \emph{obstacles}), which share only one edge
or non-edge.  We then infer that this edge or non-edge has to be
included in any editing set of size at most~$k$, and hence we can
perform the necessary edit and decrement the budget.

In Section~\ref{sec:modulator} we perform a greedy algorithm that
iteratively packs disjoint induced~$C_4$s and~$P_4$s in the graph.
Note that if we are able to pack more than~$k$ of them, then this
certifies that the considered instance does not have a solution, and
we can terminate the algorithm.  Hence, if~$X$ is the union of vertex
sets of the packed obstacles, then~$|X|\leq 4k$ and~$G-X$ is a
trivially perfect graph.  Uncovering such a set~$X$, which we call a
\modulator{}, imposes a lot of structure on the considered instance,
and is the key for further analysis of irrelevant parts of the input.

Although the applied modulator technique is standard in the area of
kernelization for graph modification problems, in this paper we
introduce a new twist to it that may have possible further
applications.  Namely, we observe that since we consider edge editing
problems, the packed obstacles do not have to be entirely
vertex-disjoint, but the next obstacle can be packed even if it shares
one vertex with the union of vertex sets of the previous obstacles; In
some limited cases even having two vertices in common is permitted.
Thus, the obtained modulator~$X$ has the property that not only is
there no obstacle in the graph~$G$ that is vertex-disjoint with~$X$,
but even the existence of obstacles sharing one vertex with~$X$ is
forbidden.  This simple observation enables us to reason about the
adjacency structure between~$X$ and $V(G)\setminus X$.  In
Section~\ref{sec:modulator-nei} we analyze this structure in order to
prove the most important technical result of the proof: The number of
subsets of~$X$ that are neighborhoods within~$X$ of vertices from
$V(G)\setminus X$ is bounded polynomially in~$k$; see
Lemma~\ref{lem:bounded-x-neighborhoods}.

In Section~\ref{sec:impbags} we proceed to analyze the trivially
perfect graph $G-X$.  Having the polynomial bound on the number of
neighborhoods within $X$, we can locate in the UCD of $G-X$ a
polynomial (in~$k$) number of \emph{important bags}, where something
interesting from the point of view of $X$-neighborhoods happens.  The
parts between the important bags have very simple structure.  They are
either \emph{tassels}: sets of trees hanging below some important bag,
where each such tree is a module in the whole graph $G$; or
\emph{combs}: long paths stretched between two important bags where
all the vertices of subtrees attached to the path have exactly the
same neighborhood in $X$.  Tassels and combs are treated differently:
Large tassels contain large trivially perfect modules in $G$ that can
be reduced quite easily, however for combs we need to devise a quite
complicated irrelevant vertex rule that locates a vertex that can be
safely discarded in a long comb.  The module reduction rules are
described in Section~\ref{sec:module-rules}, while in
Section~\ref{sec:kernel-final} we reduce the sizes of tassels and
combs and conclude the proof.


\subsection{Basic rules}\label{sec:basic-rules}

In this section we introduce the first two basic reduction rules.  In
the argumentation of the next sections, we will assume that none of
these rules is applicable.  An instance satisfying this property will
be called \emph{reduced}.

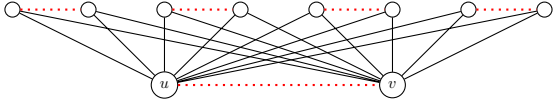
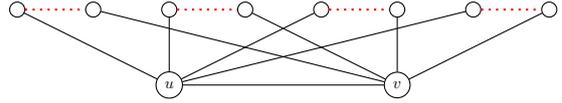
\begin{figure}[t]
  \centering
  \begin{subfigure}[t]{0.45\textwidth}
    \centering
    \begin{tikzpicture}[every node/.style={circle, draw, scale=.6},
      scale=1]
      \node (u) at (2,0) {$u$};
      \node (v) at (5,0) {$v$};

      \node (a1) at (0,1) {};
      \node (a2) at (1,1) {};

      \node (b1) at (2,1) {};
      \node (b2) at (3,1) {};

      \node (c1) at (4,1) {};
      \node (c2) at (5,1) {};

      \node (d1) at (6,1) {};
      \node (d2) at (7,1) {};

      \draw[color=red, dotted, thick] (u) -- (v);
      
      \draw (u) -- (a1) -- (v) -- (a2) -- (u);
      \draw[color=red, dotted, thick] (a1) -- (a2);
      
      \draw (u) -- (b1) -- (v) -- (b2) -- (u);
      \draw[color=red, dotted, thick] (b1) -- (b2);

      \draw (u) -- (c1) -- (v) -- (c2) -- (u);
      \draw[color=red, dotted, thick] (c1) -- (c2);

      \draw (u) -- (d1) -- (v) -- (d2) -- (u);
      \draw[color=red, dotted, thick] (d1) -- (d2);
    \end{tikzpicture}
    \caption{Rule~\ref{rule:c4}: There are four $C_4$s sharing only the
      vertices $u$ and $v$. Unless the edge $uv$ is added, we must use
      at least as many edits as the size of the non-matching.}
  \end{subfigure}
  \hspace{.05\textwidth}
  \begin{subfigure}[t]{0.45\textwidth}
    \centering
    \begin{tikzpicture}[every node/.style={circle, draw, scale=.6},
      scale=1]
      \node (u) at (2,0) {$u$};
      \node (v) at (5,0) {$v$};

      \node (a1) at (0,1) {};
      \node (a2) at (1,1) {};

      \node (b1) at (2,1) {};
      \node (b2) at (3,1) {};

      \node (c1) at (4,1) {};
      \node (c2) at (5,1) {};

      \node (d1) at (6,1) {};
      \node (d2) at (7,1) {};

      \draw (u) -- (v);
      
      \draw (u) -- (a1);
      \draw (u) -- (b1);
      \draw (u) -- (c1);
      \draw (u) -- (d1);

      \draw (v) -- (a2);
      \draw (v) -- (b2);
      \draw (v) -- (c2);
      \draw (v) -- (d2);
      
      \draw[color=red, dotted, thick] (a1) -- (a2);
      \draw[color=red, dotted, thick] (b1) -- (b2);
      \draw[color=red, dotted, thick] (c1) -- (c2);
      \draw[color=red, dotted, thick] (d1) -- (d2);
    \end{tikzpicture}
    \caption{Rule~\ref{rule:p4}: There are four $P_4$s sharing only the
      vertices $u$ and $v$.  Unless the edge $uv$ is deleted, we must
      use at least as many edits as the size of the non-matching.}
  \end{subfigure}
  \caption{Illustrations of Rules~\ref{rule:c4} and~\ref{rule:p4}.  The
    red dotted edges are non-edges; They form a matching in the
    complement graph.  In each of the cases, the only common vertices
    are $u$ and $v$.}
  \label{fig:rules-1-2}
\end{figure}

\begin{redrule}
  \label{rule:c4}
  For an instance $(G,k)$ with $uv \notin E(G)$, if there is a matching of
  size at least $k+1$ in $\overline{G[N(u) \cap N(v)]}$, then add edge
  $uv$ to~$G$ and decrease~$k$ by one, i.e., return the new instance
  $(G+uv, k-1)$.
\end{redrule}
\begin{redrule}
  \label{rule:p4}
  For an instance $(G,k)$ with $uv \in E(G)$ and $N_1 = N(u) \setminus N[v]$ and
  $N_2 = N(v) \setminus N[u]$, if there is a matching in
  $\overline{G}$ between $N_1$ and $N_2$ of size at least $k+1$, then
  delete edge $uv$ from $G$ and decrease~$k$ by one, i.e., return the
  new instance $(G-uv,k-1)$.
\end{redrule}

\begin{lemma}
  \label{lem:basic-rules-sound}
  Applicability of Rules~\ref{rule:c4} and~\ref{rule:p4} can be
  recognized in polynomial time.  Moreover, both these rules are safe,
  i.e., the input instance $(G,k)$ is a yes-instance if and only if
  the output instance $(G',k-1)$ is a yes-instance.
\end{lemma}
\begin{proof}
  Observe that verifying applicability of Rule~\ref{rule:c4} or
  \ref{rule:p4} to a fixed (non-)edge~$uv$ boils down to computing the
  cardinality of the maximum matching in an auxiliary graph.  This
  problem is well-known to be solvable in polynomial
  time~\cite{edmonds1965paths}.  Thus, by iterating over all edges and
  non-edges of~$G$ we obtain polynomial time algorithms for
  recognizing applicability of Rules~\ref{rule:c4} and~\ref{rule:p4}.
  We proceed to the proof of the safeness for both rules.
  
  \textit{Rule~\ref{rule:c4}}: Let $x_0y_0,x_1y_1,\ldots,x_{k}y_{k}$
  be edges of the found matching in $\overline{G[N(u) \cap N(v)]}$.
  Observe that for each $i$, $0\leq i\leq k$, vertices $u,x_i,v,y_i$
  induce a~$C_4$ in~$G$.  These induced~$C_4$s share only the
  non-edge~$uv$, hence any editing set that does not contain~$uv$ must
  contain at least one element of $\binom{\{u,x_i,v,y_i\}}{2}\setminus
  \{uv\}$, and consequently be of size at least~$k+1$.  We infer that
  every editing set for~$G$ that has size at most~$k$ has to include
  the edge~$uv$, and the safeness of the rule follows.


  \textit{Rule~\ref{rule:p4}}: We proceed similarly as for
  Rule~\ref{rule:c4}.  Suppose $x_0y_0,x_1y_1,\ldots,x_{k}y_{k}$ is
  the found matching in~$\overline{G}$, where $x_i\in N_1$ and $y_i\in
  N_2$ for $0\leq i\leq k$.  Then vertices $x_i,u,v,y_i$ induce
  a~$P_4$, and all these~$P_4$s for $0\leq i\leq k$ pairwise share
  only the edge~$uv$.  Similarly as for Rule~\ref{rule:c4}, we
  conclude that every editing set for~$G$ of size at most~$k$ has to
  contain~$uv$, and the safeness of the rule follows.
\end{proof}



We can now use Lemma~\ref{lem:basic-rules-sound} to apply
Rules~\ref{rule:c4} and~\ref{rule:p4} exhaustively; note that each
application reduces the budget $k$, hence at most $k$ applications can
be performed before discarding the instance as a no-instance.  From
now on, we assume that the considered instance $(G,k)$ is reduced.

\subsection{Modulator construction}\label{sec:modulator}
We now move to the construction of a small modulator whose
\textit{raison d'être} is to expose structure in the considered
graph~$G$.  We say that a subset $W\subseteq V(G)$ with $|W|=4$ is an
\emph{obstruction} if $G[W]$ is isomorphic to a $C_4$ or a $P_4$.
Formally, our modulator will be compliant to the following definition.

\begin{definition}[\modulator]\label{def:modulator}
  Let~$(G,k)$ be an instance of \TPE.  A subset~$X\subseteq V(G)$ is a
  \emph{\modulator{}} if for every obstruction $W$ the following holds
  (see Figure~\ref{fig:forbidden-modulator}):
  \begin{itemize}
  \item $|W\cap X|\geq 2$, and
  \item if $|W\cap X|=2$, then it cannot happen that $G[W]$ is a $C_4$ of
    the form $x_1-y_1-y_2-x_2-x_1$ or a $P_4$ of the form
    $x_1-y_1-y_2-x_2$, where $W\cap X=\{x_1,x_2\}$.
  \end{itemize}
  We call a \modulator{} $X$ {\emph{small}} if $|X|\leq 4k$.
\end{definition}

In particular, observe that for a \modulator{} $X$ there is no
obstacle disjoint with $X$, so $G-X$ is trivially perfect.  The
following result shows that from now we can assume that a small
\modulator{} is given to us.

\begin{lemma}
  \label{lem:polytime-modulator}
  Given an instance $(G,k)$ for \TPE{}, we can in polynomial time
  construct a small \modulator{} $X\subseteq V(G)$, or correctly
  conclude that $(G,k)$ is a no-instance.
\end{lemma}
\begin{proof}
  The algorithm starts with $X_0=\emptyset$, and iteratively constructs an
  increasing family of sets $X_0\subseteq X_1\subseteq X_2\subseteq
  \ldots$.  In the~$i$th iteration we look for an obstacle~$W$ that
  contradicts the fact that~$X_{i-1}$ is a \modulator{} according to
  Definition~\ref{def:modulator}, by verifying all the quadruples of
  vertices in~$O(n^4)$ time.  If this check verifies that~$X_{i-1}$ is
  a \modulator, then we terminate the algorithm and output $X =
  X_{i-1}$.  Otherwise, we set $X_{i} = X_{i-1}\cup W$ and proceed to
  the next iteration.  Moreover, if we performed~$k+1$ iterations,
  i.e., successfully constructed set~$X_{k+1}$, then we terminate the
  algorithm concluding that~$(G,k)$ is a no-instance.  Since in each
  iteration the next~$X_i$ grows by at most~$4$ vertices, we infer
  that if we succeed in outputting a \modulator{}~$X$, then it has
  size at most~$4k$.
  
  We are left with proving that if the algorithm successfully
  constructed $X_{k+1}$, then $(G,k)$ is a no-instance.  To this end,
  we prove by induction on $i$ that for every $i=0,1,\ldots,k+1$ and
  every editing set $F$ for $G$, it holds that $|F\cap
  \binom{X_i}{2}|\geq i$.  Indeed, from this statement for $i=k+1$ we
  can infer that every editing set for $G$ has size at least $k+1$, so
  $(G,k)$ is a no-instance.  The base of the induction is trivial, so
  for the induction step suppose that $X_i=X_{i-1}\cup W$, where $W$
  is an obstacle with $|W\cap X_{i-1}|\leq 1$ or having the form
  described in the second point of Definition~\ref{def:modulator}.
  
  First, if $|W\cap X_{i-1}|\leq 1$, then $\binom{W}{2}$ is disjoint with
  $\binom{X_{i-1}}{2}$.  Since $F$ is an editing set for $G$, we have
  that $F\cap \binom{W}{2}\neq \emptyset$, and hence
  \[
  \left|F \cap \binom{X_i}{2}\right|
  \geq \left|F \cap \binom{X_{i-1}}{2}\right| + \left|F\cap \binom{W}{2}\right|
  \geq i-1+1=i , 
  \]
  by the induction hypothesis.  Second, if $|W\cap X_{i-1}|=2$ and $W$
  has one of the two forms described in the second point of
  Definition~\ref{def:modulator}, then it is easy to see that $F$ in
  fact has to have a nonempty intersection with $\binom{W}{2}\setminus
  \{x_1x_2\}$: editing only the (non)edge $x_1x_2$ would turn a $C_4$
  into a $P_4$ or vice versa.  Since $\binom{W}{2}\setminus
  \{x_1x_2\}$ is disjoint with $\binom{X_{i-1}}{2}$, we analogously
  obtain that
  \[
  \left|F \cap \binom{X_i}{2}\right| \geq \left|F \cap
    \binom{X_{i-1}}{2}\right| + \left|F \cap
    \left(\binom{W}{2}\setminus \{x_1x_2\}\right)\right| \geq i-1+1=i
  .
  \]
\end{proof}

By applying Lemma~\ref{lem:polytime-modulator}, from now on we assume
that we are given a small \modulator{} $X$ in $G$.

\subsection{Bounding the number of neighborhoods in a \modulator{}}
\label{sec:modulator-nei}

Recall that we exposed a small \modulator{}~$X$ in the input
graph~$G$.  In polynomial time we compute the universal clique
decomposition $\mathcal{T} = (T,\mathcal{B})$ of the trivially perfect
graph~$G-X$.  The goal of this section is to analyze the structure of
neighborhoods within~$X$ of vertices residing outside~$X$.

\begin{definition}[$X$-neighborhood]
  Let $G$ be a graph and $X \subseteq V(G)$.  For a vertex $v \in V(G) \setminus
  X$, the \emph{$X$-neighborhood} of $v$, denoted $N^X_G(v)$, is the
  set $N_G(v) \cap X$.  The family of $X$-neighborhoods of $G$ is the
  set $\{N^X_G(v)\ \colon\ v \in V(G) \setminus X\}$.
  
\end{definition}

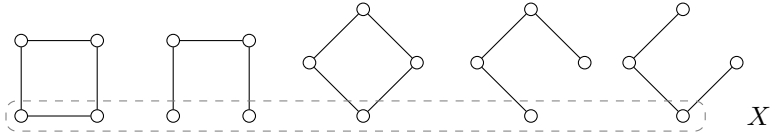
\begin{figure}[t]
  \centering
  \begin{tikzpicture}[every node/.style={circle, draw, scale=.5},
    scale=1]
    
    \def\hs{0.7071} 
    
    \node (a1) at (0,0) {};
    \node (a2) at (1,0) {};
    \node (a3) at (0,1) {};
    \node (a4) at (1,1) {};
    
    \draw (a1) -- (a2);
    \draw (a2) -- (a4);
    \draw (a3) -- (a4);
    \draw (a1) -- (a3);

    \node (b1) at (2,0) {};
    \node (b2) at (3,0) {};
    \node (b3) at (2,1) {};
    \node (b4) at (3,1) {};
    
    \draw (b1) -- (b3);
    \draw (b4) -- (b3);
    \draw (b2) -- (b4);
    
    \node (c1) at (4.5,\hs+\hs) {};
    \node (c2) at (4.5-\hs,\hs) {};
    \node (c3) at (4.5,0) {};
    \node (c4) at (4.5+\hs,\hs) {};
    
    \draw (c1) -- (c2);
    \draw (c2) -- (c3);
    \draw (c3) -- (c4);
    \draw (c4) -- (c1);
    
    \node (d1) at (6+\hs,0) {};
    \node (d2) at (6,\hs) {};
    \node (d3) at (6+\hs,\hs+\hs) {};
    \node (d4) at (6+\hs+\hs,\hs) {};
    
    \draw (d1) -- (d2);
    \draw (d2) -- (d3);
    \draw (d3) -- (d4);

    \node (e1) at (8+\hs,0) {};
    \node (e2) at (8,\hs) {};
    \node (e3) at (8+\hs,\hs+\hs) {};
    \node (e4) at (8+\hs+\hs,\hs) {};
    
    \draw (e1) -- (e2);
    \draw (e2) -- (e3);
    \draw (e1) -- (e4);
    
    \draw[dashed, gray, rounded corners] (-.2, -.2) rectangle (9, .2);
    \node[draw=none,scale=2] (x) at (9+\hs,0) {$X$};
    
  \end{tikzpicture}
  \caption{Forbidden patterns of intersection between an obstruction and
    a \modulator{} $X$.}
  \label{fig:forbidden-modulator}
\end{figure}


Again, we shall omit the subscript~$G$ whenever this does not lead to
any confusion.  Recall that the UCD~$\mathcal{T}$ gives us a
quasi-ordering~$\preceq$ on the vertices of~$G-X$.
We have~$u \preceq v$ if the bag to which~$v$ belong is a descendant
of the bag which~$u$ belongs to, where every bag is considered its own
descendant.
We shall use the notation~$u \prec v$ to denote that~$u\preceq v$
and~$v\npreceq u$.  The following two lemmas show that the
quasi-ordering~$\preceq$ is compatible with the inclusion ordering of
$X$-neighborhoods.

\begin{lemma}\label{lem:X-nei-inclusion-diff}
  If $u \prec v$ then $N^X(u) \supseteq N^X(v)$.
\end{lemma}
\begin{proof}
  Suppose $u\in B_t$ and $v\in B_s$, where $t\neq s$ and $t$ is an ancestor
  of $s$ in the forest $T$.  Recall that in a UCD, every non-leaf node
  has at least two children, which means that there exists some node
  $s'$ that is a descendant of $t$, but which is incomparable with
  $s$.  Let $w$ be any vertex of $B_{s'}$.  From the definition of a
  UCD it follows that $uv,uw\in E(G)$ but $vw\notin E(G)$.
  
  For the sake of contradiction suppose that $N^X(u)\not\supseteq
  N^X(v)$, which means there exists a vertex~$x\in X$ with~$xv\in
  E(G)$ and~$xu\notin E(G)$.  It follows that~$\{x,u,v,w\}$ is an
  obstacle regardless of whether~$wx$ is an edge or a non-edge: it is
  an induced~$C_4$ if~$wx\in E(G)$ and an induced~$P_4$ if~$wx\notin
  E(G)$.  Thus we have uncovered an obstacle sharing only one vertex
  with~$X$, contradicting the fact that~$X$ is a \modulator.
\end{proof}

\begin{lemma}\label{lem:X-nei-inclusion-same}
  If $u,v\in B_t$ for some $B_t\in \mathcal{B}$, then $N^X(u)\subseteq N^X(v)$ or
  $N^X(v)\subseteq N^X(u)$.
\end{lemma}
\begin{proof}
  Since $u,v\in B_t$, we have that $uv\in E(G)$.  For the sake of
  contradiction, suppose that there exist some $x_u \in N^X(u)
  \setminus N^X(v)$ and $x_v \in N^X(v) \setminus N^X(u)$.  It can be
  now easily seen that regardless whether $x_ux_v$ belongs to $E(G)$
  or not, the quadruple $\{u,v,x_u,x_v\}$ forms one of the obstacles
  forbidden in the second point of the Definition~\ref{def:modulator}.
  This is a contradiction with the fact that $X$ is a \modulator.
\end{proof}

Lemmas~\ref{lem:X-nei-inclusion-diff}
and~\ref{lem:X-nei-inclusion-same} motivate the following refinement
of the quasi-ordering $\preceq$: If $u,v$ belong to different bags of
$\mathcal{T}$, then we put $u\preceq_N v$ if and only if $u\preceq v$,
and if they are in the same bag, then $u\preceq_N v$ if and only if
$N^X(u)\supseteq N^X(v)$.  Thus, by
Lemma~\ref{lem:X-nei-inclusion-same} $\preceq_N$ refines $\preceq$ by
possibly splitting every bag of $\mathcal{T}$ into a family of
linearly ordered equivalence classes.  Moreover, by
Lemmas~\ref{lem:X-nei-inclusion-diff}
and~\ref{lem:X-nei-inclusion-same} we have the following corollary.

\begin{corollary}\label{cor:preceqN}
  If $u\preceq_N v$ then $N^X(u)\supseteq N^X(v)$.
\end{corollary}

Observe that for a pair of vertices $u,v\in V(G)\setminus X$, the following
conditions are equivalent: (a)~$u$ and $v$ are comparable
w.r.t~$\preceq$, (b)~$u$ and $v$ are comparable w.r.t.~$\preceq_N$,
and (c)~$uv\in E(G)$.  We have now prepared all the tools needed to
prove the main lemma from this section.

\begin{lemma}
  \label{lem:bounded-x-neighborhoods}
  If $(G,k)$ is a reduced instance for \TPE{} and~$X$ is a small
  \modulator, then the number of different~$X$-neighborhoods is at
  most~$O(k^4)$.
\end{lemma}
\begin{proof}
  Let $\mathcal{F}$ be the family of $X$-neighborhoods in $G$.  For
  every $Z\in \mathcal{F}$, let us choose an arbitrary vertex $v_Z\in
  V(G)\setminus X$ with $Z=N^X(v_Z)$.  We split $\mathcal{F}$ into two
  subfamilies:
  The first family~$\mathcal{F}_1$ contains all the sets
  of~$\mathcal{F}$ that contain the endpoints of some non-edge
  in~$G[X]$,
  whereas
  the second family~$\mathcal{F}_2$ contains all the sets
  of~$\mathcal{F}$ that induce complete graphs in~$G[X]$.  We bound
  the sizes of~$\mathcal{F}_1$ and~$\mathcal{F}_2$ separately.
  
  %
  %
  \medskip
  
  \noindent\emph{Bounding $|\mathcal{F}_1|$}: Let $xy$ be a non-edge
  of~$G[X]$, and for $2 \leq \kappa \leq |X|$ let $\mathcal{F}_1^{xy,\kappa}$ be the
  family of those sets of~$\mathcal{F}_1$ that contain~$\{x,y\}$ and
  have cardinality exactly $\kappa$.  Take any distinct $Z_1,Z_2\in
  \mathcal{F}_1^{xy,\kappa}$, and observe that they are not nested
  since both have size~$\kappa$.  By Corollary~\ref{cor:preceqN}, this
  means that vertices~$v_{Z_1}$ and~$v_{Z_2}$ are incomparable
  w.r.t.~$\preceq_N$, so $v_{Z_1}v_{Z_2}\notin E(G)$.  Hence, set
  $\{v_Z\colon Z\in \mathcal{F}_1^{xy,\kappa}\}$ is independent
  in~$G$.  Observe now that if we had that $|\{v_Z\colon Z\in
  \mathcal{F}_1^{xy,\kappa}\}|\geq 2k+2$, then Rule~\ref{rule:c4}
  would be applicable to the non-edge~$xy$.  Since we assume that the
  instance is reduced, we conclude that $|\{v_Z\colon Z\in
  \mathcal{F}_1^{xy,\kappa}\}|\leq 2k+1$, and hence also
  $|\mathcal{F}_1^{xy,\kappa}|\leq 2k+1$.  By summing through all
  the~$\kappa$ between~$2$ and~$|X|$ and through all the non-edges
  of~$G[X]$, we infer that
  \[
  |\mathcal{F}_1|
  \leq \binom{4k}{2}\cdot 4k \cdot (2k+1) 
  = O(k^4) .
  \]

  %
  %
  \medskip
  
  \noindent \emph{Bounding $|\mathcal{F}_2|$}: Consider any pair of
  $X$-neighborhoods $Z_1,Z_2\in \mathcal{F}_2$ such that they are not
  nested, and moreover there exist vertices $x_1 \in Z_1 \setminus
  Z_2$ and $x_2 \in Z_2 \setminus Z_1$ such that $x_1x_2 \in E(G)$.
  Since~$Z_1$ and~$Z_2$ are not nested, by Corollary~\ref{cor:preceqN}
  we infer that~$v_{Z_1}$ and~$v_{Z_2}$ are incomparable
  w.r.t.~$\preceq_N$, and hence $v_{Z_1}v_{Z_2}\notin E(G)$.  Observe
  that then $G[\{v_{Z_1},v_{Z_2},x_1,x_2\}]$ is an induced~$P_4$;
  however, the existence of such an obstacle is not forbidden by the
  definition of a \modulator.
  
  Create an auxiliary graph~$H$ with~$V(H)=\mathcal{F}_2$, and put
  $Z_1Z_2\in E(H)$ if and only if~$Z_1$ and~$Z_2$ satisfy the
  condition from the previous paragraph, i.e.,~$Z_1$ and~$Z_2$ are not
  nested and there exist $x_1 \in Z_1 \setminus Z_2$ and $x_2 \in Z_2
  \setminus Z_1$ with $x_1x_2 \in E(G)$.  Run the classic greedy
  $2$-approximation algorithm for vertex cover in~$H$.  This algorithm
  either finds a matching~$M$ in~$H$ of size more than
  $\binom{4k}{2}\cdot k$, or a vertex cover~$C$ of~$H$ of size at most
  $2\cdot \binom{4k}{2}\cdot k$.  In the first case, assign each
  edge~$Z_1Z_2$ of~$M$ to the corresponding edge~$x_1x_2$ of~$G[X]$ as
  in the definition of the edges of~$H$.  Observe that since~$|X|\leq
  4k$, then some edge~$x_1x_2\in G[X]$ is assigned at least~$k+1$
  times.  Then it is easy to see that the sets
  $\{v_{Z_1},v_{Z_2},x_1,x_2\}$ for $Z_1Z_2$ being edges of~$M$
  assigned to~$x_1x_2$ induce~$P_4$s that share only the
  edge~$x_1x_2$, and hence Rule~\ref{rule:p4} would be applicable
  to~$x_1x_2$.  This is a contradiction with the assumption
  that~$(G,k)$ is reduced.  Hence, we can assume that we have
  successfully constructed a vertex cover~$C$ of~$H$ of size at most
  $2\cdot \binom{4k}{2}\cdot k=O(k^3)$.
  
  Let now $\mathcal{F}_2'=\mathcal{F}_2\setminus C$.  Since~$\mathcal{F}_2'$
  is independent in~$H$, it follows that for any non-nested
  $Z_1,Z_2\in \mathcal{F}_2'$ and any $x_1\in Z_1\setminus Z_2$,
  $x_2\in Z_2\setminus Z_1$, we have that $x_1x_2\notin E(G)$.  Since
  the sets of~$\mathcal{F}_2'$ induce complete graphs in~$G[X]$, this
  means that in particular there is no set $Z_3\in \mathcal{F}_2'$
  that contains both~$x_1$ and~$x_2$.  This proves that the
  family~$\mathcal{F}_2'$ is a \tpss{} with~$X$ as ground set, so by
  Lemma~\ref{lem:tpss-bounded} we infer that $|\mathcal{F}_2'| \leq |X|+1
  \leq 4k+1$.  Concluding,
  \[
  |\mathcal{F}_2|
  \leq |C| + |\mathcal{F}_2'|
  \leq O(k^3) + 4k + 1
  = O(k^3) ,
  \]
  and $|\mathcal{F}| \leq |\mathcal{F}_1| + |\mathcal{F}_2| = O(k^4) +
  O(k^3) = O(k^4)$.
\end{proof}

\newcommand{\lca}{lowest common ancestor-closure}

\subsection{Locating important bags}
\label{sec:impbags}

\begin{figure}[t]
  \centering
  \begin{subfigure}[t]{.25\textwidth}
    \centering
    \includegraphics[width=.75\textwidth]{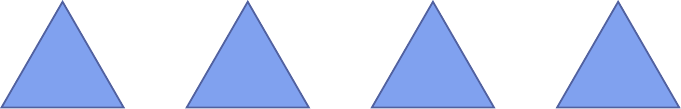}
    \caption{Type~0: $x$ sees a disjoint union of connected components.}
  \end{subfigure}
  \hspace{2em}
  \begin{subfigure}[t]{.25\textwidth}
    \centering
    \includegraphics[width=.75\textwidth]{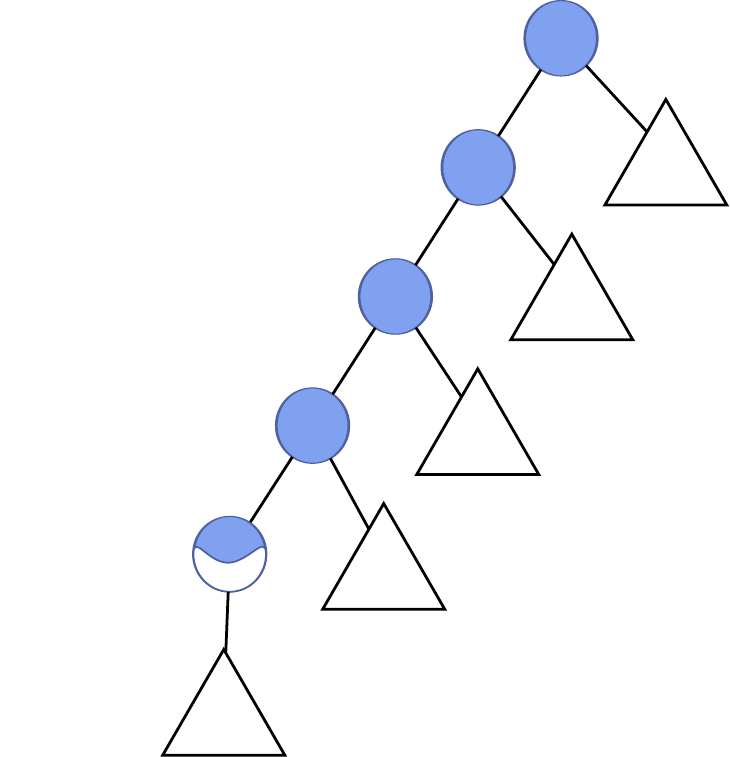}
    \caption{Type~1: $x$ sees all the vertices in bags from a root and to
      a point in a bag, and nothing else.}
  \end{subfigure}
  \hspace{2em}
  \begin{subfigure}[t]{.25\textwidth}
    \centering
    \includegraphics[width=.75\textwidth]{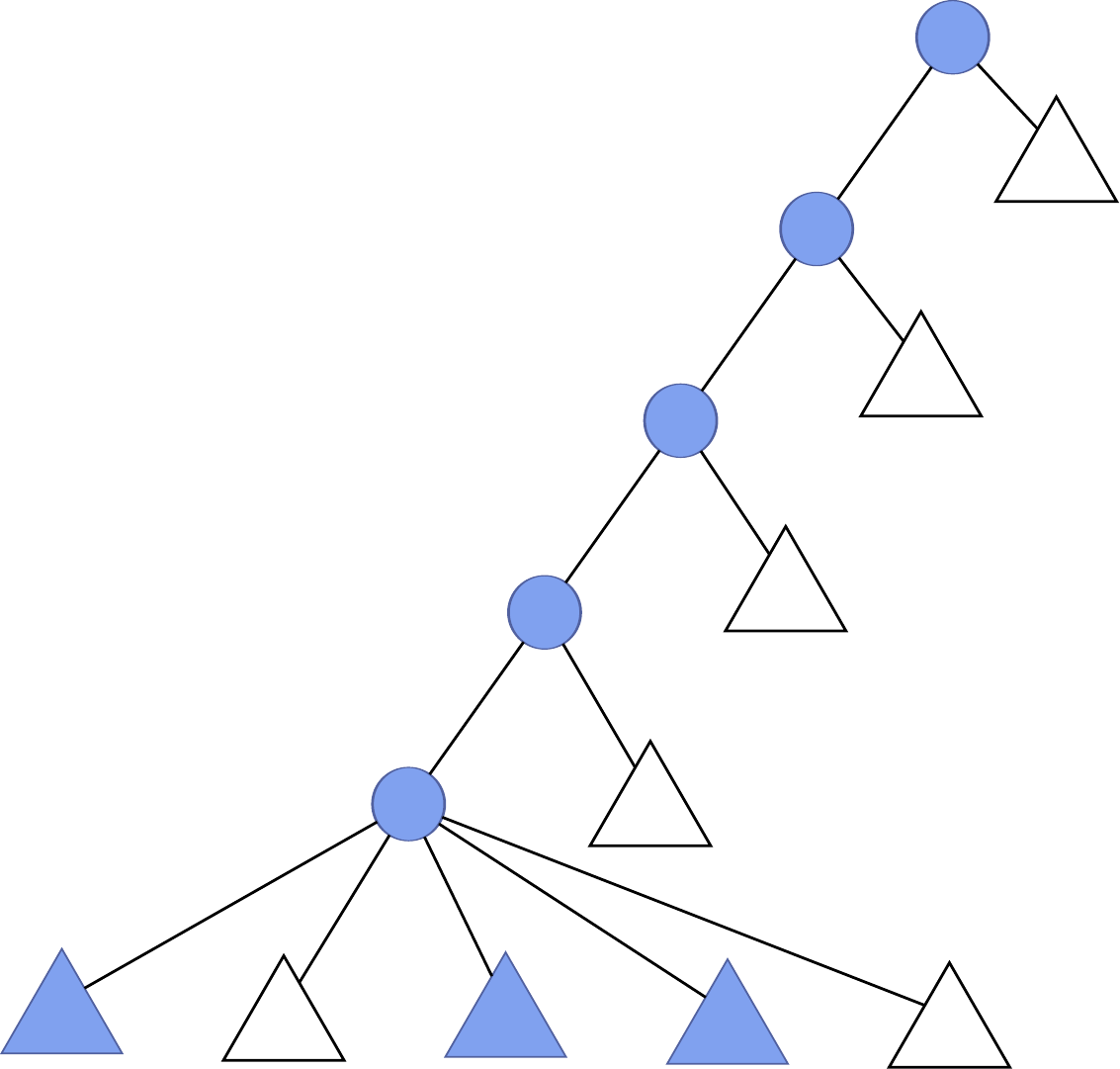}
    \caption{Type~2: $x$ has as neighbors all the vertices from a root
      and down to a bag, and a collection of subtrees below that bag.}
  \end{subfigure}
  \caption{Three types of neighborhoods; simply denoted Type~0, Type~1,
    and Type~2.  The blue parts mark the possible neighborhoods of a
    vertex $x \in X$.}
  \label{fig:nx}
\end{figure}

In the previous section we analyzed the structure of neighborhoods
that nodes from~$V(G)\setminus X$ have in~$X$.  Our goal in this
section is to perform the symmetric analysis: to understand, how the
neighborhood of a fixed~$x\in X$ in~$V(G)\setminus X$ looks like.
Eventually, we aim to locate a family~$I$ of~$O(k)$ \emph{important
  bags}, where some non-trivial behavior w.r.t.~the neighborhoods of
vertices of~$X$ happens.
Then, we will perform a \lca{} on the set~$I$, thus increasing its
size to at most twice.  After performing this step, all the connected
components of~$T-I$ have very simple structure from the point of view
of their neighborhoods in~$X$.  As there are only~$O(k)$ such
components, we will be able to kernelize them separately.

The following definition and lemma explains what are the types of
neighborhoods that vertices of~$X$ can have in~$V(G)\setminus X$.  To
simplify the notation, in the following we treat~$\preceq$ also as a
partial order on the vertices of the forest~$T$ denoting the
ancestor-descendant relation, i.e.,~$s\preceq t$ if and only if~$s$ is
an ancestor of~$t$ (possibly~$s=t$).

\begin{definition}[Type~0,~1, and~2 neighborhoods]
  Let $x \in X$ be any vertex and consider $U_x = N(x) \setminus X$.  We say
  that $U_x$ is (see Figure~\ref{fig:nx}):
  \begin{itemize}
  \item A \emph{neighborhood of Type~0} if $U_x$ is the union of the
    vertex sets of a collection of connected components of $G-X$.
  \item A \emph{neighborhood of Type~1} if there exists a node $t_x\in
    V(T)$ such that $\bigcup_{s\prec t_x} B_s \subseteq U_x\subseteq
    \bigcup_{s\preceq t_x} B_s$.  In other words, $U_x$ consists of
    all the vertices contained in bags on the path from $t_x$ to the
    root of its subtree in $T$, where some vertices of $B_{t_x}$
    itself may be excluded.
  \item A \emph{neighborhood of Type~2} if there exists a node $t_x\in
    V(T)$ and a collection $\mathcal{L}_x$ of subtrees of $T$ rooted
    at children of $t_x$ such that $U_x=\bigcup_{s\preceq t_x} B_s
    \cup \bigcup_{S\in \mathcal{L}_x}\bigcup_{s\in V(S)} B_s$.  In
    other words, $U_x$ is formed by all the vertices contained in bags
    on the path from $t_x$ to the root of its subtree in $T$, plus a
    selection of subtrees rooted in the children of $t_x$, where the
    vertices appearing in the bags of each such subtree are either all
    included in $U_x$ or all excluded from $U_x$.
  \end{itemize}
\end{definition}

\begin{lemma}
  \label{lem:nx-two-types}
  Let $x \in X$ be any vertex and consider $U_x = N(x) \setminus X$.  Then $U_x$
  is of Type~0,~1 or~2.
\end{lemma}
\begin{proof}
  From Corollary~\ref{cor:preceqN} we infer that $U_x$ is closed
  downwards w.r.t.~the quasi-ordering $\preceq_N$, i.e., if $v\in U_x$
  and $u\preceq_N v$, then also $u\in U_x$.  Let $S_x$ be the set of
  nodes of $T$ whose bags contain at least one vertex of $U_x$.  It
  follows that $S_x$ is closed under taking ancestors in forest $T$.
  Moreover if $t\in S_x$, then the bags of all the ancestors of $t$
  other than $t$ are fully contained in $U_x$.
  
  \begin{claim}\label{cl:types}
    Suppose $t,t'\in S_x$ are two nodes that are incomparable
    w.r.t.~$\preceq$.  Then $U_x \supseteq \bigcup_{s\succeq t} B_s$
    and $U_x\supseteq \bigcup_{s\succeq t'} B_s$, i.e., $U_x$ contains
    all the vertices of all the bags contained in the subtrees of $T$
    rooted at $t$ and $t'$.
  \end{claim}
  \begin{proof}
    We prove the statement for the subtree rooted at~$t'$; The proof
    for the subtree rooted at~$t$ is symmetric.  Let~$y$ and~$y'$ be
    arbitrary vertices of~$B_t\cap U_x$ and~$B_{t'}\cap U_x$,
    respectively.  For the sake of contradiction suppose there exists
    some~$v\in \bigcup_{s\succeq t'} B_s$ such that~$vx\notin E(G)$.
    Since~$v\in \bigcup_{s\succeq t'} B_s$ and~$t,t'$ are incomparable
    w.r.t.~$\preceq$, by the properties of the universal clique
    decomposition we have that~$yy'\notin E(G)$, $vy\notin E(G)$ and
    $vy'\in E(G)$.  Since $xy,xy'\in E(G)$ by the definition of~$U_x$,
    we conclude that~$\{y,y',x,v\}$ would induce a~$P_4$ in~$G$ that
    has only one vertex in common with~$X$ (see
    Figure~\ref{fig:broken-type}), a contradiction to the definition
    of a \modulator.  \cqed\end{proof}
  
  \begin{figure}[htp]
    \centering
    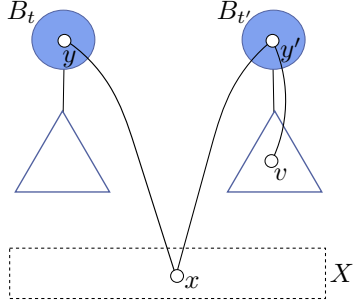
    \caption{An induced $P_4$, $y x y' v$, with only one vertex~$x$ in
      the modulator, appearing in the proof of Claim~\ref{cl:types}.}
    \label{fig:broken-type}
  \end{figure}
  
  We now use Claim~\ref{cl:types} to perform a case study that
  recognizes $U_x$ as a neighborhood of Type~0,~1, or~2.
  
  Suppose first that~$U_x$ contains vertices of at least two distinct
  connected components of~$G-X$.  Let~$C_1,C_2$ be any two such
  components, and let~$T_1$ and~$T_2$ be the trees of the forest~$T$
  that are UCDs of~$C_1$ and~$C_2$, respectively.  Since~$S_x$ is
  closed under taking ancestors in~$T$, it follows that the roots
  of~$T_1$ and~$T_2$ belong to~$S_x$.  Claim~\ref{cl:types} implies
  then that the entire vertex sets of~$C_1$ and~$C_2$ are contained
  in~$U_x$.  Since~$C_1,C_2$ was an arbitrary pair of components
  containing a vertex of~$U_x$, it follows that~$U_x$ must be the
  union of vertex sets of a selection of connected components
  of~$G-X$, i.e., a neighborhood of Type~0.
  
  Since~$U_x=\emptyset$ is also a neighborhood of Type~0, we are left with
  analyzing the case when $U_x\subseteq V(C_0)$ for~$C_0$ being a
  connected component of~$G-X$; Let~$T_0$ be
  the UCD of~$C_0$.  Observe that if~$U_x$ does not contain any pair
  of vertices incomparable w.r.t.~$\preceq$, then~$S_x$ must form a
  path from some node of~$T_0$ to the root of~$T_0$, and hence~$U_x$
  is a neighborhood of Type~1.  Otherwise, there exists some node
  of~$S_x$ such that at least two subtrees rooted at its children
  contain nodes from~$S_x$.  Let~$t_x$ be such a node that is highest
  in~$T_0$, and let~$\mathcal{L}_x$ be the family of subtrees rooted
  at children of~$t_x$ that contain nodes of~$S_x$.  Again applying
  Claim~\ref{cl:types}, we infer that~$U_x$ contains all the vertices
  of all the bags of every subtree of~$\mathcal{L}_x$: for any two
  distinct subtrees~$T_1,T_2\in \mathcal{L}_x$,~$S_x$ contains the
  roots of~$T_1$ and~$T_2$, and hence by Claim~\ref{cl:types}~$U_x$
  contains all the vertices of all the bags of~$T_1$ and~$T_2$.
  Since~$t_x$ was chosen to be the highest, it follows that~$U_x$ is a
  neighborhood of Type~2 for node~$t_x$ and selection of
  subtrees~$\mathcal{L}_x$.
\end{proof}

Clearly, for every $x\in X$ we can in polynomial time analyze $U_x$ and
recognize it as a neighborhood of Type~0,~1, or~2.  Let $I_0$ be the
set of nodes $t_x$ for vertices $x\in X$ for which $U_x$ is of
Type~$1$ or~$2$.  To simplify the structure of $T-I_0$, we perform the
\lca{} operation on $I_0$.  The following variant of this operation is
taken verbatim from the work of Fomin et al.~\cite{fomin2012planar}.

%
%
\begin{definition}[\cite{fomin2012planar}]
  For a rooted tree $T$ and vertex set $M \subseteq V (T )$ the \lca{}
  (LCA-closure) is obtained by the following process.  Initially, set
  $M' = M$.  Then, as long as there are vertices $x$ and $y$ in $M'$
  whose least common ancestor $w$ is not in $M'$, add $w$ to $M'$.
  When the process terminates, output $M'$ as the LCA-closure of $M$.
  The following folklore lemma summarizes two basic properties of
  LCA-closures.
\end{definition}
\begin{lemma}[\cite{fomin2012planar}]
  \label{lem:lca}
  Let $T$ be a tree, $M \subseteq V(T)$ and $M' =
  \textnormal{LCA-closure}(M)$.  Then $|M'| \leq 2|M|$ and for every
  connected component $C$ of $T-M'$, $|N(C)| \leq 2$.
\end{lemma}
%
%

Construct now the set $I$ by taking $\textnormal{LCA-closure}(I_0)$
and
%
%
adding  the root of every connected component of~$T$ that contains a bag
of~$I_0$  (provided it is not already included).
The nodes from~$I$ will be called \emph{important nodes}, or
\emph{important bags}.  From Lemma~\ref{lem:lca} it follows that~$|I|
\leq 3|X| \leq 12k$, and by the construction we infer that every
connected component~$C$ of~$T-I$ is of one of the following three
forms:
\begin{itemize}
\item $C$ is not adjacent to any node of~$I$, and is thus simply a
  connected component of~$T$ that does not contain any important bag.
\item $C$ is adjacent to one node~$a$ of~$I$, and it is a subtree
  rooted at a child of~$a$.
\item $C$ is adjacent to two nodes~$a$ and~$b$ of~$I$ such that~$a$ is
  an ancestor of~$b$.  Then~$C$ is formed by the internal nodes of
  the~$a-b$ path in~$T$, plus all the subtrees rooted at the other
  children of these internal nodes.
\end{itemize}

\subsection{Module reduction}\label{sec:module-rules}
In this section we give two new reduction rules: a twin reduction and
a module reduction rule.  These rules are executed exhaustively by the
algorithm as Rules~\ref{rule:twin} and~\ref{rule:module}.  The reason
why we introduce them now is that only after understanding the
structural results of Sections~\ref{sec:modulator-nei}
and~\ref{sec:impbags}, the motivation of these rules becomes apparent.
Namely, these rules will be our main tools in reducing the sizes of
parts of $G-X$ located between the important bags.

\subsubsection{Twin reduction}
\begin{redrule}
  \label{rule:twin}
  If~$T \subseteq V(G)$ is a true twin class of size~$|T| > 2k + 5$, and~$v
  \in T$ is an arbitrarily picked vertex, then remove~$v$ from the
  graph, i.e., proceed with the instance~$(G-v,k)$.
\end{redrule}
\begin{lemma}
  \label{lem:rule:twin:sound}
  Applicability of Rule~\ref{rule:twin} can be recognized in
  polynomial time.  Moreover, Rule~\ref{rule:twin} is safe,
  i.e.,~$(G,k)$ is a yes-instance if and only if~$(G-v,k)$ is a
  yes-instance.
\end{lemma}
\begin{proof}
  In order to recognize the applicability of Rule~\ref{rule:twin} we
  only need to inspect every true twin classes in the graph, which
  clearly can be done in polynomial time.  We proceed to the proof of
  the safeness of the rule.
  
  Let~$T$ be a true twin class of size at least~$2k+5$ and let~$v$ be
  the vertex the rule deleted.  Since the class of trivially perfect
  graphs is hereditary, if~$(G,k)$ is a yes-instance, it follows that
  $(G-v,k)$ is a yes-instance.  Suppose now that $(G-v,k)$ is a
  yes-instance.  Let~$F$ be a set of edges with~$|F| \leq k$ such that
  $(G-v) \triangle F$ is trivially perfect.  We now show that $G
  \triangle F$ is also trivially perfect, which means that~$F$ is also
  a solution to~$(G,k)$.  For the sake of contradiction, suppose~$W$
  is an obstruction in~$G \triangle F$.  Since $(G-v) \triangle F$ is
  trivially perfect,~$W$ must contain the deleted vertex~$v$.
  Since~$F$ has size at most~$k$, at most~$2k$ vertices of~$T$ can be
  incident to an edge of~$F$.  Let~$v_1$, $v_2$, $v_3$, and~$v_4$ be
  four vertices of~$T$ that are different from~$v$ and are not
  incident to~$F$.  Then one of them, say~$v_1$, is not contained
  in~$W$.  Since~$v$ and~$v_1$ are true twins both in $G$ and in
  $G\triangle F$, we can replace~$v$ with~$v_1$ in~$W$ yielding a new
  set~$W'$ which is an obstruction in $G \triangle F$.  However,
  since~$v$ is not a member of~$W'$, we have that~$W'$ is an
  obstruction in $(G-v) \triangle F$, contradicting the assumption
  that $(G-v) \triangle F$ was trivially perfect.
\end{proof}

\subsubsection{Module reduction}
Recall that a module is a set of vertices $M$ such that for every
vertex $v$ in $V(G) \setminus M$, either $M \subseteq N(v)$ or $M \cap
N(v) = \emptyset$; see Definition~\ref{def:module}.  The following
rule enables us to reduce large trivially perfect modules.

\begin{redrule}
  \label{rule:module}
  Suppose $M \subseteq V(G)$ is a module such that $G[M]$ is trivially perfect
  and it contains an independent set of size at least $2k + 5$.  Then
  let us take any independent set $I\subseteq M$ of size $2k+4$, and
  we delete every vertex of $M$ apart from $I$, i.e., proceed with the
  instance $(G-(M\setminus I),k)$.
\end{redrule}

Observe that Rule~\ref{rule:module} always deletes at least one
vertex, since $|M|\geq 2k+5$ and $|I|=2k+4$.  Actually, we could
define a stronger rule where we only assume that $|M|\geq 2k+5$;
however, the current statement will be helpful in recognizing the
applicability of Rule~\ref{rule:module}.

We first prove that the rule is indeed safe.

\begin{lemma}
  \label{lem:rule-mod-correct}
  Provided that $(G,k)$ is a reduced instance
  (w.r.t.~Rules~\ref{rule:c4} and~\ref{rule:p4}), then
  Rule~\ref{rule:module} is safe, i.e., $(G,k)$ is a yes-instance if
  and only if $(G-(M\setminus I),k)$ is a yes-instance.
\end{lemma}
\begin{proof}  
  Let $A = M \setminus I$, and $G' = G-A$.  Since~$G'$ is an induced subgraph
  of~$G$, by heredity, if $(G,k)$ is a yes-instance, then $(G', k)$ is
  a yes-instance.  We proceed to the proof of the other direction.
  Suppose then that $(G',k)$ is a yes-instance, and let~$F$, $|F|\leq k$,
  be a minimum-size editing set for~$G'$.
  \begin{claim}
    \label{claim:i-notin-f}
    No vertex of~$I$ is incident to any edit of~$F$.
  \end{claim}
  \begin{proof}
    Since~$F$ has minimum possible size, it is inclusion-wise minimal.
    We show that if $F_I \subseteq F$ is the set of edges of~$F$
    incident to a vertex of~$I$ and $F' = F \setminus F_I$, then $G'
    \triangle F$ being trivially perfect implies $G' \triangle F'$
    being trivially perfect.  Since~$|I|=2k+4$, we can find at least
    four vertices $v_1, \ldots, v_4\in I$ that are not incident to any
    edit of~$F$.  Suppose that $G' \triangle F'$ is not trivially
    perfect.  Then there is an obstruction~$W$ in $G' \triangle F'$
    containing at least one of the vertices of~$I$ incident to an edge
    of~$F$.  Create~$W'$ by replacing every vertex of $(W\cap
    I)\setminus \{v_1,\ldots,v_4\}$ by a different vertex of
    $\{v_1,\ldots,v_4\}$ that is not contained in~$W$.  Since vertices
    of~$I$ are not incident to the edits of~$F'$, they are false twins
    in $G' \triangle F'$, and hence~$W'$ created in this manner
    induces a graph isomorphic to the one induced by~$W$.  Thus,~$W'$
    is an obstacle in $G' \triangle F'$.  However, the vertices
    $v_1,\ldots,v_4$ are not incident to the edits of~$F$ and
    hence~$W'$ induces the same graph in $G' \triangle F'$ as in $G'
    \triangle F$.  Therefore~$W'$ would be an obstacle in $G'\triangle
    F$, a contradiction to $G'\triangle F$ being trivially perfect.
    
    Since we argued that $F'\subseteq F$ is also a solution, by the optimality
    of $F$ we infer that $F=F'$ and $F_I=\emptyset$.  \cqed\end{proof}
  
  We now argue that $G\triangle F$ is trivially perfect, which will
  imply that $(G,k)$ is a yes-instance.  For the sake of
  contradiction, suppose that there exists an obstacle~$W$
  in~$G\triangle F$; It follows that~$W$ shares at least one vertex
  with~$M \setminus I$.  From Claim~\ref{claim:i-notin-f} it follows
  that no edit of~$F$ is incident to any vertex of~$M$, so in~$G
  \triangle F$ we still have that~$M$ is a module.
  
  If the obstruction~$W$ induces a~$P_4$, then it is known that~$W$ is
  fully contained in the module~$M$, or has at most one vertex
  in~$M$~\cite[Observation~1]{guillemot2013onthenon}.
  Since~$G[M]=(G\triangle F)[M]$ is trivially perfect, the latter is
  the case.  But since~$M$ is a module in $G\triangle F$, then
  replacing the single vertex of $W\cap A$ with any vertex of~$I$
  would yield an obstacle in $G'\triangle F$, a contradiction.
  
  Consider then the case when~$W$ induces a~$C_4$ in $G\triangle F$.
  Since~$G[M]=(G\triangle F)[M]$ is~$C_4$-free, we have that~$W$ is
  not entirely contained in~$M$.  Also, if~$W$ had three vertices
  in~$M$, then the remaining vertex would need to be contained
  in~$N_G(M)$, and hence would be adjacent in $G\triangle F$ to all
  the other three vertices of~$W$, a contradiction to~$(G\triangle
  F)[W]$ being a~$C_4$.  Therefore, at most two vertices of~$W$ can be
  in~$M$.
  
  Suppose exactly two vertices~$w_1$ and~$w_3$ of~$W$ are in~$M$,
  and~$w_2$ and~$w_4$ are outside~$M$.  As $M$ is a module both in $G$
  and in $G\triangle F$, we must have that $w_2,w_4\in N_G(M)$ and
  hence the $4$-cycle induced by $W$ in $G\triangle F$ must be
  $w_1-w_2-w_3-w_4-w_1$.  Take any two vertices $w_1',w_3'\in I$ and
  obtain $W'$ by replacing $w_1$ and $w_3$ with them.  It follows that
  $W'$ induces a $C_4$ in $G'\triangle F$, a contradiction.
  
  Finally, consider the case when exactly one vertex of~$W$,
  say~$w_1$, is in~$M$.  Again, replacing $w_1$ with any vertex of $I$
  would yield an induced $C_4$ contained in $G'\triangle F$, a
  contradiction.  Thus, we conclude that $G\triangle F$ is trivially
  perfect.
\end{proof}

Observe that in order to apply Rule~\ref{rule:module}, one needs to be
given the module~$M$.  Given~$M$, finding any independent
set~$I\subseteq M$ of size~$2k+4$ can then be done easily as follows:
We can find an independent set of maximum cardinality in~$M$ in
polynomial time, since~$G[M]$ is trivially perfect and the
{\sc{Independent Set}} problem is polynomial-time solvable on
trivially perfect graphs (it boils down to picking one vertex from
every leaf bag of the universal clique decomposition of the considered
graph).  Then we take any of its subsets of size~$2k+4$ to be~$I$.
Hence, to apply Rule~\ref{rule:module} exhaustively, we need the
following statement.

\begin{lemma}\label{lem:recognizing-module}
  There exists a polynomial-time algorithm that, given an instance
  $(G,k)$, either finds a module $M\subseteq V(G)$ where
  Rule~\ref{rule:module} can be applied, or correctly concludes that
  Rule~\ref{rule:module} is inapplicable.
\end{lemma}
\begin{proof}
  Using Theorem~\ref{thm:module-decomp} we compute the module
  decomposition~$(T,(M^t)_{t \in V(T)})$ of~$G$.  Then we verify
  applicability of Rule~\ref{rule:module} to each module~$M^t$
  for~$t\in V(T)$, by checking whether~$G[M]$ is trivially perfect and
  contains an independent set of size~$2k+5$ (the latter check can be
  done in polynomial time since~$G[M]$ is trivially perfect).
  Moreover, we perform the same check on all the modules~$N_t$ formed
  as follows: take a union node~$t\in V(T)$, and construct a
  module~$N_t$ by taking the union of labels of those children of~$t$
  that induce trivially perfect graphs.
  
  We now argue that if Rule~\ref{rule:module} is applicable to some
  module~$M$ in~$G$, then this algorithm will encounter some (possibly
  different) module~$M'$ to which Rule~\ref{rule:module} is applicable
  as well.  By the third point of Theorem~\ref{thm:module-decomp},
  either~$M=M^t$ for some~$t\in V(T)$, or~$M$ is the union of a
  collection of labels of children of some union or join node.  In the
  first case the algorithm verifies~$M$ explicitly.  In the following,
  let~$\alpha(H)$ denote the size of a maximum independent set in a
  graph~$H$.
  
  If now~$M$ is a union of labels of some children of a union
  node~$t$, then by heredity~$M\subseteq N^t$.  Moreover,~$N^t$
  induces a trivially perfect graph (since trivially perfect graphs
  are closed under taking disjoint union) and clearly $\alpha(N^t)\geq
  \alpha(M)$.  Hence, Rule~\ref{rule:module} is applicable to
  $M'=N^t$, and this will be discovered by the algorithm.
  
  Finally, suppose~$M$ is a union of labels of some children
  $t_1,t_2,\ldots,t_p$ of a join node~$t$.  Observe that since for
  every~$i\neq j$, every vertex of~$M^{t_i}$ is adjacent to every
  vertex of~$M^{t_j}$, it follows that
  $\alpha(G[M])=\max_{i=1,2,\ldots,p}\alpha(G[M^{t_i}])$.  Without
  loss of generality suppose that the maximum on the right hand side
  is attained for the module~$M^{t_1}$.  Then by heredity~$G[M^{t_1}]$
  is trivially perfect, and $\alpha(G[M^{t_1}])=\alpha(G[M])\geq
  2k+5$.  Therefore Rule~\ref{rule:module} is applicable to
  $M'=M^{t_1}$, and this will be discovered by the algorithm.
\end{proof}

We remark here that for the kernelization algorithm it is not
necessary to be sure that Rule~\ref{rule:module} is inapplicable at
all.  Instead, we could perform it on demand.  More precisely, during
further analysis of the structure of $G-X$ we argue that some modules
have to be small, since otherwise Rule~\ref{rule:module} would be
applicable.  This analysis can be performed by a polynomial-time
algorithm that would just apply Rule~\ref{rule:module} on any
encountered module that needs shrinking.  However, we feel that the
fact that Rule~\ref{rule:module} can be indeed applied exhaustively
provides a better insight into the algorithm, and streamlines the
presentation.

Having introduced and verified Rules~\ref{rule:twin}
and~\ref{rule:module}, we can now prove that after applying them
exhaustively, all the trivially perfect modules in the graph are
small.

\begin{lemma}
  \label{lem:small-ind-twin-small}
  A (possibly disconnected) trivially perfect graph with maximum true
  twin class size $t$ and maximum independent set size $\alpha$ has at most
  $(2\alpha - 1)t$ vertices in total.
\end{lemma}
\begin{proof}
  Let $\mathcal{T}$ be the UCD of $G$, a trivially perfect graph with
  independent set number $\alpha$ and every true twin class of size at
  most $t$.  Since any collection comprising one vertex from each leaf
  bag of $\mathcal{T}$ forms an independent set, there are at most
  $\alpha$ leaf bags in $\mathcal{T}$.  Thus the number of nodes of
  $\mathcal{T}$ in total is at most $2 \alpha - 1$.  Since every bag
  of the decomposition $T \subseteq V(G)$ is a true twin class, we
  conclude that there are at most $(2\alpha - 1)t$ vertices in $G$.
\end{proof}

\begin{corollary}\label{cor:small-modules}
  Suppose an instance $(G,k)$ is reduced, and moreover
  Rules~\ref{rule:twin} and~\ref{rule:module} are not applicable to
  $(G,k)$.  Then for every module $M \subseteq V(G)$ such that $G[M]$
  is trivially perfect, we have that $|M|=O(k^2)$.
\end{corollary}
\begin{proof}
  Suppose $M$ is such a module.  Observe that members of every true
  twin class in $G[M]$ are also true twins in $G$ (since $M$ is a
  module).  Hence twin classes in $G[M]$ have size at most $2k+4$, as
  otherwise Rule~\ref{rule:twin} would be applicable.  Moreover, if
  $G[M]$ contained an independent set of size $2k+5$, then
  Rule~\ref{rule:module} would be applicable.  By
  Lemma~\ref{lem:small-ind-twin-small}, we infer that $|M|\leq
  (4k+7)(2k+4)=O(k^2)$.
\end{proof}

From now on we assume that in the considered instance $(G,k)$ we have
exhaustively applied Rules~\ref{rule:c4}--\ref{rule:module}, using the
algorithms of
Lemmas~\ref{lem:basic-rules-sound},~\ref{lem:rule:twin:sound},
and~\ref{lem:recognizing-module}.  Hence
Corollary~\ref{cor:small-modules} can be used.  Observe that to
perform this step, we do not need to construct the small modulator $X$
at all.  However, we hope that the reader already sees that
Rules~\ref{rule:c4}--\ref{rule:module} will be useful for shrinking
too large parts of $G-X$ between the important bags.

\subsection{Kernelizing non-important parts (irrelevant vertex
  deletion)}
\label{sec:kernel-final}

\newcommand{\up}{\uparrow}

\newcommand{\down}{\downarrow}

Recall that we have fixed a small \modulator{}~$X$ with~$|X|\leq 4k$ such
that~$G-X$ is a trivially perfect graph with universal clique
decomposition~$\mathcal{T}$.  Moreover,
Rules~\ref{rule:c4}--\ref{rule:module} are inapplicable to $(G,k)$.
By Lemma~\ref{lem:bounded-x-neighborhoods} we have that the number of
$X$-neighborhoods is $O(k^4)$.  By the marking procedure, we have
marked a set~$I$ of~$O(k)$ bags of~$\mathcal{T}$ as important, in such
a manner that every connected component of~$\mathcal{T}-I$ is adjacent
to at most two vertices of~$I$, and is in fact of one of the three
forms described at the end of Section~\ref{sec:impbags}.

Thus, the whole vertex set of $G-X$ can be partitioned into four sets:
\begin{description}
\item[$V_I$:] vertices contained in bags from~$I$;
\item[$V_0$:] vertices contained in bags of those components
  of~$\mathcal{T}-I$ that are not adjacent to any bag from~$I$;
\item[$V_1$:] vertices contained in bags of those components
  of~$\mathcal{T}-I$ that are adjacent to exactly one bag from~$I$;
\item[$V_2$:] vertices contained in bags of those components
  of~$\mathcal{T}-I$ that are adjacent to exactly two bags from $I$.
\end{description}
We are going to establish an upper bound on the cardinality of each of
these sets separately.  Upper bounds for~$V_I$,~$V_0$, and~$V_1$
follow already from the introduced reduction rules, but for~$V_2$ we
shall need a new reduction rule.  The upper bounds on the
cardinalities of~$V_I$ and~$V_0$ are quite straightforward.

\begin{lemma}\label{lem:vI}
  $|V_I|\leq O(k^6)$.
\end{lemma}
\begin{proof}
  Consider for some $a \in I$ the bag $B_a$.  Note that $B_a$ is a
  module in $G-X$.  By Lemma~\ref{lem:bounded-x-neighborhoods} there
  are only $O(k^4)$ possible $X$-neighborhoods among vertices of
  $G-X$.  Hence, vertices of~$B_a$ can be partitioned into~$O(k^4)$
  classes w.r.t.~the neighborhoods in~$X$.  Each such class is a
  module in~$G$ that is also a clique, and hence it is a true twin
  class.  Since the twin reduction rule (Rule~\ref{rule:twin}) is not
  applicable, each true twin class has size at most~$2k+5$, which
  implies that $|B_a|\leq O(k^5)$.  As $|I|=O(k)$, we conclude that
  $|V_I|\leq O(k^6)$.
\end{proof}

We remark that using a more precise analysis of the situation in one
bag $B_a$ for $a\in I$, one can see that the $X$-neighborhoods of
elements of $B_a$ are nested, so there is only at most $|X|+1\leq
4k+1$ of them.  By plugging in this argument in the proof of
Lemma~\ref{lem:vI}, we obtain a sharper upper bound of $O(k^3)$
instead of $O(k^6)$.  However, the upper bounds on $|V_0|$ and $|V_1|$
are $O(k^6)$ and $O(k^7)$, respectively, so establishing a better
bound here would have no influence on the overall asymptotic kernel
size.  Hence, we resorted to a simpler proof of a weaker upper bound.

\begin{lemma}\label{lem:v0}
  $|V_0|\leq O(k^6)$.
\end{lemma}
\begin{proof}
  Observe that $V_0$ is the union of bags of these connected
  components of $G-X$, whose universal clique decompositions (being
  components of $\mathcal{T}$) do not contain any important bag.  By
  the definition of important bags, each such connected component~$C$
  is a module in~$G$, and clearly its neighborhood is entirely
  contained in~$X$.  Recall that by
  Lemma~\ref{lem:bounded-x-neighborhoods} there are only~$O(k^4)$
  possible different~$X$-neighborhoods among vertices of~$G-X$.  Thus,
  we can group the connected components of~$G[V_0]$ according to their
  $X$-neighborhoods into~$O(k^4)$ groups, and the union of vertex sets
  in each such group forms a module in~$G$.  Since
  Rule~\ref{rule:module} is not applicable, by
  Corollary~\ref{cor:small-modules} we have that each of these modules
  has size~$O(k^2)$.  Thus we infer that~$|V_0|\leq O(k^6)$.
\end{proof}

To bound the size of $V_1$ we need a few more definitions.  Suppose
that~$C$ is a component of $\mathcal{T}-I$ that is adjacent to exactly
one important bag $a\in I$.  By the construction of~$I$, we have
that~$C$ is a tree rooted in a child of~$a$.  We shall say that~$C$ is
\emph{attached below~$a$}.  The union of bags of all the components of
$\mathcal{T}-I$ attached below~$a$ will be called the \emph{tassel
  rooted at~$a$}.  Thus,~$V_1$ can be partitioned into~$O(k)$ tassels.

\begin{lemma}\label{lem:one-tassel}
  For every $a\in I$, the tassel rooted at $a$ has size at most
  $O(k^6)$.
\end{lemma}
\begin{proof}
  Let $C_1,C_2,\ldots,C_r$ be the components of $\mathcal{T}-I$ rooted
  at the children of $a$, whose union of bags forms the tassel rooted
  at $a$.  Recall that none of the $C_i$s contains any important bag.
  Therefore, from Lemma~\ref{lem:nx-two-types} we infer that for any
  $C_i$ and any $x\in X$, either all the vertices from the bags of
  $C_i$ are adjacent to $x$, or none of them.  Thus, the union of bags
  of each $C_i$ forms a module in $G$: The vertices in this union have
  the same $X$-neighborhood, and moreover their neighborhoods in $G-X$
  are formed by the vertices from the bags on the path from $a$ to the
  root of $a$'s connected component in $\mathcal{T}$.  Similarly as in
  the proof of Lemma~\ref{lem:v0}, by
  Lemma~\ref{lem:bounded-x-neighborhoods} there are only $O(k^4)$
  possible $X$-neighborhoods, so we can partition the components $C_i$
  into $O(k^4)$ classes with respect to their neighborhoods in $X$.
  The union of bags in each such class forms a module in $G$; since
  Rule~\ref{rule:module} is not applicable, by
  Corollary~\ref{cor:small-modules} we infer that its size is bounded
  by $O(k^2)$.  Thus, the total number of vertices in all the
  components $C_i$ is at most $O(k^6)$.
\end{proof}

As $|I|=O(k)$, Lemma~\ref{lem:one-tassel} immediately implies the
following.

\begin{lemma}\label{lem:v1}
  $|V_1|\leq O(k^7)$.
\end{lemma}

We are left with bounding the cardinality of $V_2$.  Let us fix any
component~$C$ of~$\mathcal{T}-I$ which is adjacent in~$\mathcal{T}$ to
two nodes of~$I$.  From the construction of $I$, it follows that $C$
has the following form:
\begin{itemize}
\item $C$ contains a path $P=a_1-a_2-\ldots-a_d$ such that in
  $\mathcal{T}$, node $a_d$ is a child of an important node $b^\up$,
  and $a_1$ has exactly one important child $b^\down$.
\item For every $i=1,2,\ldots,d$, $C$ contains also all the subtrees
  of $\mathcal{T}$ rooted in children of $a_i$ that are different from
  $a_{i-1}$ (where $a_{0}=b^\down$).
\end{itemize}
Such a component~$C$ will be called a \emph{comb} (see
Figure~\ref{fig:comb}).  The path~$P$ is called the \emph{shaft} of a
comb; the union of the bags of the shaft will be denoted by~$Q$.  The
union of the bags of the subtrees rooted in children of~$a_i$, apart
from~$a_{i-1}$, will be called the \emph{tooth at~$i$}, and denoted
by~$R_i$.  Note that the subgraph induced by a tooth is not
necessarily connected; it is, however, always non-empty by the
definition of the universal clique decomposition.  We also denote
$R=\bigcup_{i=1}^d R_i$.  By somehow abusing the notation, we will
also denote $B_i=B_{a_i}$ for $i=1,2,\ldots,d$.  The number of
teeth~$d$ is called the \emph{length} of a comb.

%
%

\begin{figure}[t]
  \centering
  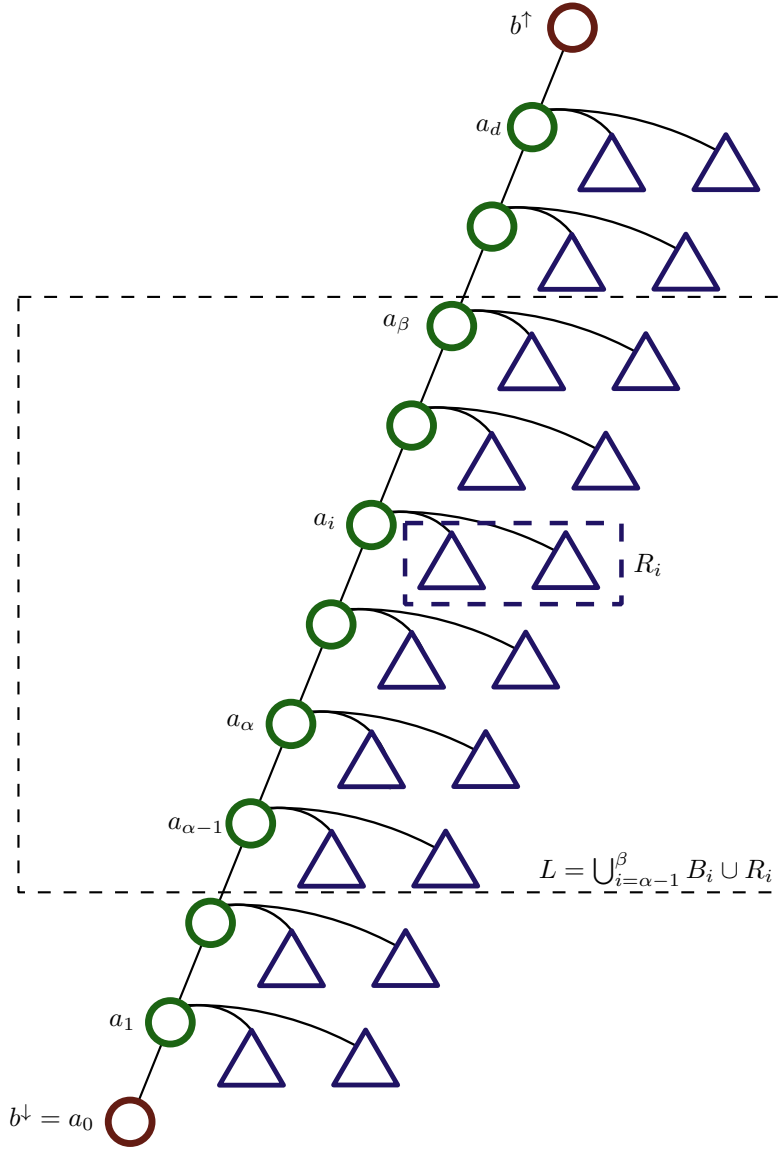
  \caption{The anatomy of a \emph{comb}.  The top and bottom bags,
    $b^\up$ and $b^\down$, are important bags.}
  \label{fig:comb}
\end{figure}

%
%

Since the comb $C$ does not contain any important vertices, from
Lemma~\ref{lem:nx-two-types} and the construction of~$I$ we
immediately infer the following observation about the
$X$-neighborhoods of vertices of the shaft and the teeth.

\begin{lemma}\label{lem:comb-neighborhoods}
  There exists two sets $Y,Z$ with
  $Z \subseteq Y \subseteq X$
  such that $N_X(u)=Y$ for every $u\in Q$ and $N_X(v)=Z$ for every $v\in
  R$.
\end{lemma}

In particular, Lemma~\ref{lem:comb-neighborhoods} implies that every
tooth of a comb is a module.  Hence, since Rule~\ref{rule:module} is
not applicable, we infer that $|R_i|=O(k^2)$ for $i=1,2,\ldots,d$.
Also, observe that each $B_i$ is a twin class, so by inapplicability
of Rule~\ref{rule:twin} we conclude that $|B_i|\leq 2k+5$ for each
$i=1,2,\ldots,d$.

Since $\mathcal{T}$ is a forest and $|I|=O(k)$, it follows that in
$\mathcal{T}-I$ there are $O(k)$ combs.  As we already observed, for
each comb the sizes of individual teeth and bags on the shaft are
bounded polynomially in $k$.  Hence, the only thing that remains is to
show how to reduce combs that are long.  In order to do this, we need
one more definition: a tooth $R_i$ is called \emph{simple} if $G[R_i]$
is edgeless, and it is called \emph{complicated} otherwise.  We can
now state the final reduction rule.

\newcommand{\rmv}{\beta} \newcommand{\rmvt}{R_{\rmv}}
\newcommand{\lind}{\alpha}

\begin{redrule}\label{rule:comb}
  Suppose $C$ is a comb of length at least $(4k+3)^2$, and adopt the
  introduced notation for the shaft and the teeth of $C$.  Define an
  index $\rmv$ as follows:
  \begin{enumerate}[(i)]
  \item\label{c1} If at least $4k+3$ teeth $R_i$ are complicated, then
    we let $\rmv=d$.
  \item\label{c2} Otherwise, there is a sequence of $4k+3$ consecutive
    teeth $R_i,R_{i+1},\ldots,R_{i+4k+2}$ that are simple.  Let $\rmv$
    be the index of the last tooth of this sequence, i.e.,
    $\rmv=i+4k+2$.
  \end{enumerate}
  Having defined $\rmv$, remove the tooth $\rmvt$ from the graph and
  do not modify the budget.  That is, proceed with the instance
  $(G-\rmvt,k)$.
\end{redrule}

\begin{lemma}
  \label{lem:rule-comb-correct}
  Rule~\ref{rule:comb} is safe.
\end{lemma}
\begin{proof}
  Since $G-\rmvt$ is an induced subgraph of $G$, then we trivially
  have that the existence of a solution for $(G,k)$ implies the
  existence of a solution for $(G-\rmvt,k)$.  Hence, we now prove the
  converse.  Suppose that $F$ is a solution to $(G-\rmvt,k)$, that is,
  a set of edits in $G-\rmvt$ such that $(G-\rmvt)\triangle F$ is
  trivially perfect and $|F|\leq k$.
  
  We will say that a tooth $R_i$ is \emph{spoiled} if any vertex of
  $R_i\cup B_i$ is incident to an edit from $F$, and \emph{clean}
  otherwise.  The first goal is to find an index $\lind$ such that
  \begin{enumerate}[(a)]
  \item\label{q1} $1<\lind<\rmv$,
  \item\label{q2} the teeth $R_{\lind-1}$ and $R_{\lind}$ are clean,
    and
  \item\label{q3} if any of the teeth $R_{\lind+1},
    R_{\lind+2},\ldots,\rmvt$ is complicated, then $R_{\lind}$ is also
    complicated.
  \end{enumerate}
  
  Suppose first that $\rmv$ was constructed according to case
  (\ref{c1}), i.e., there are at least $4k+3$ complicated teeth in the
  comb, and hence $\rmv=d$.  Out of these teeth $R_i$, at most one can
  have index $1$, at most one can have index $d$, at most $2k$ can be
  spoiled (since $|F|\leq k$) and at most $2k$ can have the preceding
  tooth $R_{i-1}$ spoiled.  This leaves at least one complicated tooth
  $R_i$ such that $1<i<d$ and both $R_i$ and $R_{i-1}$ are clean.
  Then we can take $\lind=i$; thus, property (\ref{q3}) of $\lind$ is
  satisfied since $R_{\lind}$ is complicated.
  
  Suppose then that $\rmv$ was constructed according to case
  (\ref{c2}), i.e.,  the following teeth are all simple:
  $R_{\rmv-(4k+2)}, R_{\rmv-(4k+1)}, \ldots, R_{\rmv-1},\rmvt$.
  Similarly as before, out of these $4k+3$ teeth, one has index
  $\rmv$, one has index $\rmv-(4k+2)$, at most $2k$ can be spoiled,
  and at most $2k$ can have the preceding tooth spoiled.  Hence, among
  them there is a tooth $R_i$ such that $\rmv-(4k+2)<i<\rmv$ and both
  $R_i$ and $R_{i-1}$ are clean.  Again, we take $\lind=i$; thus,
  property (\ref{q3}) is satisfied since all the teeth
  $R_{\rmv-(4k+2)}, R_{\rmv-(4k+1)},\ldots R_{\rmv-1},\rmvt$ are
  simple.
  
  \bigskip

  %
  %
  With $\lind$ defined, we are ready to complete the proof of
  Lemma~\ref{lem:rule-comb-correct}.  To that aim, define
  $L=\bigcup_{i=\lind-1}^{\rmv} B_i\cup R_i$.
  Construct $F'$ from $F$ by removing all the edits that are incident
  to any vertex of~$L$; clearly~$|F'|\leq |F|\leq k$.  We claim
  that~$F'$ is a solution to the instance~$(G,k)$, that is, that
  $G\triangle F'$ is trivially perfect.  For the sake of a
  contradiction, suppose that~$A\subseteq V(G)$ is a vertex set of
  size~$4$ such that~$G\triangle F'[A]$ is a~$P_4$ or a~$C_4$.  Let
  $A_0=A\cap L$ and $A_1=A\setminus A_0$.
  
  \begin{claim}\label{cl:smallA0}
    $|A_0|=1$ or $|A_0|=2$.
  \end{claim}
  \begin{proof}
    Suppose first that $A_0=\emptyset$, so $A\subseteq V(G)\setminus L\subseteq V(G-\rmvt)$.  Since
    $F\cap \binom{V(G)\setminus L}{2}=F'\cap \binom{V(G)\setminus
      L}{2}$ and $\rmvt\subseteq L$, we have that the induced subgraph
    $G\triangle F'[A]$ is equal to the induced subgraph
    $(G-\rmvt)\triangle F[A]$.  However, the graph $(G-\rmvt)\triangle
    F$ is trivially perfect, so it cannot have an induced $P_4$ or
    $C_4$; a contradiction.
    
    Suppose now that $|A_0|\geq 3$.  Since $A_0\subseteq L$ and no edit of $F'$
    is incident to any vertex of $L$, we infer that there is no edit
    of $F'$ between vertices of $A$: only at most one vertex of $A$
    does not belong to $A_0$.  Therefore $G[A]=G\triangle F'[A]$ and
    $G[A]$ is an induced $C_4$ or $P_4$ in the graph $G$.  However,
    $A_0\subseteq L\subseteq V(G)\setminus X$, so $|A\cap X|\leq 1$.
    Thus, $G[A]$ would be an obstacle in $G$ that has at most one
    common vertex with \modulator{} $X$, a contradiction with the
    definition of a \modulator{} (Definition~\ref{def:modulator}).
    \cqed\end{proof}
  
  To obtain a contradiction, we shall construct a set $A_0'$
  satisfying the following properties:
  \begin{enumerate}[(i)]
  \item\label{p1} $A_0'\subseteq R_{\lind-1}\cup B_{\lind-1}\cup R_{\lind}\cup
    B_{\lind}$;
  \item\label{p2} $|A_0'|=|A_0|$ and $G[A_0']$ is edgeless if and only
    if $G[A_0]$ is edgeless;
  \item\label{p3} $|A_0\cap Q|=|A_0'\cap Q|$ and hence $|A_0\cap R|=|A_0'\cap
    R|$.
  \end{enumerate}
  Let us define $A' = A_1\cup A_0'$.  For now we postpone the exact
  construction
  
  \begin{claim}\label{cl:iso}
    If $A_0'$ satisfies properties (\ref{p1}), (\ref{p2}), and
    (\ref{p3}), then $G\triangle F'[A]$ is isomorphic to $G\triangle F'[A']$.
  \end{claim}
  \begin{proof}
    By property (\ref{p3}) there exists a bijection~$\eta$ between~$A_0$
    and~$A_0'$ that preserves belonging to~$Q$ or~$R$ between the
    argument and the image.  Extend~$\eta$ to~$A$ by defining
    $\eta(u)=u$ for $u\in A_1$; we claim that~$\eta$ is an isomorphism
    between $G\triangle F'[A]$ and $G\triangle F'[A']$.  To see this,
    observe that since $A_0,A_0'\subseteq L$, then we have that no
    vertex of~$A_0$ or~$A_0'$ is incident to any edit of~$F'$.
    Moreover, in~$G$, all the vertices of $L\cap R$ have the same
    neighborhood in $V(G)\setminus L$, and the same holds also for the
    vertices of~$L\cap Q$.  As the neighborhoods of these vertices
    in~$G$ and in~$G\triangle F'$ are exactly the same, we infer that
    each vertex~$u\in A_0$ is adjacent in $G\triangle F'$ to the same
    vertices of~$A_1$ as the vertex~$\eta(u)$ is.
    
    To conclude the proof, we need to prove that $\eta$ restricted to
    $A_0'$ is also an isomorphism between $G\triangle F'[A_0]$ and
    $G\triangle F'[A_0']$.  Again, $A_0$ and $A_0'$ are not incident
    to any edit of $F'$, so $G\triangle F'[A_0]=G[A_0]$ and
    $G\triangle F'[A_0']=G[A_0']$.  By Claim~\ref{cl:smallA0} we have
    that $|A_0|=1$ or $|A_0|=2$, and we conclude by observing that a
    pair of simple graphs with at most two vertices are isomorphic if
    and only if both of them are edgeless or both of them contain an
    edge, and in both cases any bijection between the vertex sets is
    an isomorphism.  \cqed\end{proof}
  
  We now argue that the existence of a set $A_0'$ satisfying
  properties (\ref{p1}), (\ref{p2}), and (\ref{p3}) leads to a
  contradiction.  Recall that the teeth~$R_{\lind-1}$ and~$R_{\lind}$
  are clean, which means that no vertex of $R_{\lind-1}\cup
  B_{\lind-1}\cup R_{\lind}\cup B_{\lind}$ is incident to any edit
  from~$F$.  Moreover, as~$\rmv>\lind$, we have that $A'\subseteq
  V(G-\rmvt)$.  By the construction of~$F'$ and~$A'$ we infer that
  $G\triangle F'[A']=(G-\rmvt)\triangle F[A']$.  By Claim~\ref{cl:iso}
  we have that~$G\triangle F'[A']$ is a~$P_4$ or a~$C_4$, since
  $G\triangle F'[A]$ was.  This would, however, mean that
  $(G-\rmvt)\triangle F$ would contain an induced~$P_4$ or an
  induced~$C_4$, a contradiction to the assumption that
  $(G-\rmvt)\triangle F$ is trivially perfect.
  
  Therefore, we are left with constructing a set~$A_0'$ satisfying
  properties (\ref{p1}), (\ref{p2}), and (\ref{p3}).  We give
  different constructions depending on the alignment of the vertices
  of~$A_0$.  In each case we just define $A_0'$; verifying properties
  (\ref{p1}), (\ref{p2}), and (\ref{p3}) in each case is trivial.
  
  \begin{description}
  \item[Case 1.] $|A_0|=1$.
    \begin{description}
    \item[Case 1a.] $A_0=\{u\}$ and $u\in Q$.  Then $A_0'=\{u'\}$ for
      any $u'\in B_{\lind-1}$.
    \item[Case 1b.] $A_0=\{u\}$ and $u\in R$.  Then $A_0'=\{u'\}$ for
      any $u'\in R_{\lind-1}$.
    \end{description}
  \item[Case 2.] $|A_0|=2$.
    \begin{description}
    \item[Case 2a.] $A_0=\{u,v\}$, $u,v\in Q$.  As $G[Q]$ is a clique,
      it follows that $uv\in E(G)$.  Then $A_0'=\{u',v'\}$ for any
      $u'\in B_{\lind-1}$ and $v'\in B_{\lind}$.
    \item[Case 2b.] $A_0=\{u,v\}$, $u\in Q$, $v\in R$, and $uv\notin
      E(G)$.  Then $A_0'=\{u',v'\}$ for any $u'\in B_{\lind-1}$ and
      $v'\in R_{\lind}$.
    \item[Case 2c.] $A_0=\{u,v\}$, $u\in Q$, $v\in R$, and $uv\in E(G)$.
      Then $A_0'=\{u',v'\}$ for any $u'\in B_{\lind}$ and $v'\in
      R_{\lind-1}$.
    \item[Case 2d.] $A_0=\{u,v\}$, $u,v\in R$, and $uv\notin E(G)$.  Then
      $A_0'=\{u',v'\}$ for any $u'\in R_{\lind}$ and $v'\in
      R_{\lind-1}$.
    \item[Case 2e.] $A_0=\{u,v\}$, $u,v\in R$, and $uv\in E(G)$.  As there
      are no edges in~$G$ between different teeth, we observe that
      $u,v\in R_i$ for some~$i$ such that~$R_i\subseteq L$, i.e.,
      $\lind-1\leq i\leq \rmv$.  In particular, the tooth~$R_i$ must
      be complicated.  If~$i=\lind-1$ or~$i=\lind$, then we can take
      $A_0'=A_0$.  Otherwise we have that $\lind<i\leq \rmv$ and~$R_i$
      is complicated, so by property (\ref{q3}) of~$\rmv$ we infer
      that~$R_{\lind}$ is also complicated.  Then we take
      $A_0'=\{u',v'\}$ for any $u',v'\in R_{\lind}$ such that $u'v'\in
      E(G)$.
    \end{description}
  \end{description}

  This case study is exhaustive due to Claim~\ref{cl:smallA0}.
\end{proof}

We can finally gather all the pieces and prove our main theorem.

\begin{theorem}
  \label{thm:tpe-polykernel}
  The problem \TPE{} admits a proper kernel with $O(k^7)$ vertices.
\end{theorem}
\begin{proof}
  The algorithm first applies Reduction Rules
  \ref{rule:c4}---\ref{rule:module} exhaustively.  As each application
  of a reduction rule either decreases~$n$ and does not change~$k$, or
  decreases~$k$ while not changing~$n$, the number of applications of
  these rules will be bounded by $O(n+k)$ until~$k$ becomes negative
  and we can conclude that we are working with a \noinstance.  By
  Lemmas~\ref{lem:basic-rules-sound},
  \ref{lem:rule:twin:sound},~\ref{lem:rule-mod-correct},
  and~\ref{lem:recognizing-module}, these rules are safe,
  applicability of each rule can be recognized in polynomial time, and
  applying the rules also takes polynomial time.
  
  After Rules~\ref{rule:c4}--\ref{rule:module} have been applied
  exhaustively, we construct a small \modulator{}~$X$ using the
  algorithm of Lemma~\ref{lem:polytime-modulator}.  In case the
  construction fails, we conclude that we are working with a
  \noinstance.  Otherwise, in polynomial time we construct the
  universal clique decomposition $\mathcal{T}$ of $G-X$, and then we
  mark the set~$I$ of important bags.  Both locating the important
  bags and performing the lowest common ancestor closure can be done
  in polynomial time.  After this, we examine all the combs of
  $\mathcal{T}-I$.  In case there is a comb of length greater than
  $(4k+3)^2$, we apply Rule~\ref{rule:comb} on it and restart the
  whole algorithm.  Observe that each application of this rule reduces
  the vertex count by one while keeping~$k$, so the total number of
  times the algorithm is restarted is bounded by the vertex count of
  the original instance.
  
  We are left with analyzing the situation when Reduction
  Rule~\ref{rule:comb} is not applicable, i.e., all the combs have
  length less than $(4k+3)^2$.  As we have argued, the inapplicability
  of Rules~\ref{rule:twin} and~\ref{rule:module} ensures that bags of
  shafts of combs have sizes~$O(k)$ and teeth of combs have
  sizes~$O(k^2)$.  Hence, every comb has~$O(k^4)$ vertices.  Since the
  number of combs is~$O(k)$, we infer that $|V_2|\leq O(k^5)$.
  Together with the upper bounds on the sizes of~$V_I$, $V_0$,
  and~$V_1$ given by Lemmas~\ref{lem:vI},~\ref{lem:v0},
  and~\ref{lem:v1}, we conclude that
  \[
  |V(G)| = |X|+|V_I|+|V_0|+|V_1|+|V_2| \leq
  4k+O(k^6)+O(k^6)+O(k^7)+O(k^5) = O(k^7) .
  \]
  Hence, we can output the current instance as the obtained kernel.
\end{proof}

\section{Kernels for \pname{Trivially Perfect Completion/Deletion}}
\label{sec:kernel-comp-del}

We now present how the technique applied to \TPE{} also yields
polynomial kernels for \TPC{} and \TPD{} after minor modifications.
That is, we prove Theorems~\ref{thm:tpd-polykernel-intro}
and~\ref{thm:tpc-polykernel-intro}.

We show that all the rules given above, with only two minor
modifications are correct for both problems.  Clearly, the running
times of the algorithms recognizing applicability of the rule do not
depend on the problem we are solving, so we only need to argue for
their safeness.

In the first two rules, Rules~\ref{rule:c4} and~\ref{rule:p4}, we add
and delete an edge, respectively, and the argument is that any editing
set of size at most~$k$ must necessarily include this edit.  However,
in the completion and deletion version, we are not allowed both
operations.  Hence, for the first rule, in the deletion variant we can
immediately infer that we are working with a no-instance, and
respectively for the second rule in the completion variant.

Thus, the two following rules replace Rule~\ref{rule:c4} for deletion
and Rule~\ref{rule:p4} for completion, and their safeness is
guaranteed by a trivial modification of the proof of
Lemma~\ref{lem:basic-rules-sound}:
\begin{customrdrl}{\ref{rule:c4}D}
  \label{rule:c4d}
  For an instance $(G,k)$ with $uv \notin E(G)$, if there is a matching of
  size at least $k+1$ in $\overline{G[N(u) \cap N(v)]}$, then return a
  trivial no-instance as the computed kernel.
\end{customrdrl}
\begin{customrdrl}{\ref{rule:p4}C}
  \label{rule:p4c}
  For an instance $(G,k)$ with $uv \in E(G)$ and $N_1 = N(u) \setminus N[v]$ and
  $N_2 = N(v) \setminus N[u]$, if there is a matching in
  $\overline{G}$ between $N_1$ and $N_2$ of size at least $k+1$, then
  return a trivial no-instance as the computed kernel.
\end{customrdrl}

Observe that Rules~\ref{rule:c4d} and~\ref{rule:p4c} are applicable in
exactly the same instances as their unmodified variants.  Hence,
exhaustive application of the basic rules with any of these
modifications results in exactly the same notion of a reduced instance
as the one introduced in Section~\ref{sec:basic-rules}.  We now argue
that Rules~\ref{rule:twin} and~\ref{rule:module} are safe for both the
deletion and the completion variant, without any modifications.

\begin{lemma}
  \label{lem:module-rules-safe}
  Rules~\ref{rule:twin} and \ref{rule:module} are safe both for \TPD{}
  and for \TPC.
\end{lemma}

\begin{proof}
  The proof of the safeness of Rule~\ref{rule:twin}
  (Lemma~\ref{lem:rule:twin:sound}) in fact argues that every editing
  set $F$ for $(G-v,k)$ with $|F|\leq k$ is also an editing set for
  $(G,k)$.  This holds also for editing sets that consist only of edge
  additions/deletions, so the reasoning remains the same for \TPD{}
  and \TPC.
  
  The proof of the safeness of Rule~\ref{rule:module}
  (Lemma~\ref{lem:rule-mod-correct}) first argues that any
  minimum-size editing set $F$ for the reduced instance $(G',k)$ is
  not incident to any vertex of $I$.  This is done by showing that
  otherwise $F$ would not be an inclusion-wise minimal editing set
  (proof of Claim~\ref{claim:i-notin-f}), and the argumentation can be
  in the same manner applied to minimum-size completion/deletion sets.
  Then it is argued that $F$ is in fact an editing set for the
  original instance $(G,k)$, and the argumentation is oblivious to
  whether $F$ is allowed to contain edge additions or deletions.
\end{proof}

We now proceed to the analysis of Rule~\ref{rule:comb} in the
completion and deletion variants.  First, let us consider the
construction of the modulator.  In the completion/deletion variants we
can construct the modulator in exactly the same manner as for editing.
Indeed, the main argument for the bound $|X|\leq 4k$ states that if
the construction was performed for more than $k$ rounds, then we are
dealing with a no-instance, since then any editing set for $G$ has
size at least $k+1$.  Completion and deletion sets are editing sets in
particular, so the same argument holds also for \tpd{} and \tpc.

Results of Sections~\ref{sec:modulator-nei} and~\ref{sec:impbags},
i.e., the analysis of the $X$-neighborhoods and marking of the
important bags, work in exactly the same manner, since they are based
on the same notions of a reduced instance and of a \modulator.  Thus,
Lemma~\ref{lem:bounded-x-neighborhoods} holds as well, and we have
marked the same set $I$ of $O(k)$ important bags, with the same
properties.  Rules~\ref{rule:twin} and~\ref{rule:module} are not
modified, so the bounds on $|V_I|$, $|V_0|$ and $|V_1|$ from
Lemmas~\ref{lem:vI},~\ref{lem:v0}, and~\ref{lem:v1} also hold.

We are left with analyzing Rule~\ref{rule:comb}, and we claim that
this rule is also safe for \tpd{} and \tpc{} without any
modifications.  Indeed, in the proof of the safeness of the rule
(Lemma~\ref{lem:rule-comb-correct}), we have argued that for every
editing set~$F$ ($|F|\leq k$) for the new instance~$(G',k)$, there
exists some~$F'\subseteq F$ which is a solution to the original
instance~$(G,k)$.  In case~$F$ consists of edge deletions or edge
additions only, so does~$F'$.  Hence,~$(G',k)$ being a yes-instance of
\tpd{}, resp.\ \tpc{}, implies that~$(G,k)$ is also a yes-instance of
the same problem.  Thus Rule~\ref{rule:comb} is safe without any
modifications, and the kernel size analysis contained in the proof of
Theorem~\ref{thm:tpe-polykernel} (end of
Section~\ref{sec:kernel-final}) can be performed in exactly the same
manner.  This concludes the proof of
Theorems~\ref{thm:tpd-polykernel-intro}
and~\ref{thm:tpc-polykernel-intro}.

\section{Hardness results}\label{sec:hardness}

In this section we show that \tpe{} is \cclass{NP}-hard, and
furthermore not solvable in subexponential parameterized time unless
the Exponential Time Hypothesis fails.  Recall that the
\cclass{NP}-hardness of the problem was already established by Nastos
and Gao~\cite{nastos2013familial}.  Their reduction (see the proof of
Theorem 3.3 in~\cite{nastos2013familial}) starts with an instance of
\pname{Exact 3-Cover} with universe of size~$n$ and set family of
size~$m$, and constructs an instance~$(G,k)$ of \tpe{} with
$k=\Theta(mn^2)$.  Thus, the parameter blow-up is at least cubic, and
the reduction cannot be used to establish the non-existence of a
subexponential parameterized algorithm under ETH.

Here, we give a direct, linear reduction from \pname{3Sat} to \tpe{}.
Furthermore, the resulting graph in our reduction has maximum degree
equal to~$4$.  Thus, we in fact prove that even on input graphs of
maximum degree~$4$, \tpe{} remains \cclass{NP}-hard and does not admit
a subexponential parameterized algorithm, unless ETH fails.  Formally,
the following theorem will be proved, where for an input
formula~$\phi$ of \pname{3Sat}, by $\mathcal{V}(\phi)$ and
$\mathcal{C}(\phi)$ we denote the variable and clause sets of~$\phi$,
respectively:

\begin{theorem}
  \label{thm:np-hardness}
  There exists a polynomial-time reduction that, given an instance $\phi$
  of \pname{3Sat}, returns an equivalent instance $(G_\phi,k_\phi)$ of
  \tpe, where $|V(G_\phi)|=13|\mathcal{C}(\phi)|$, $|E(G_\phi)| =
  18|\mathcal{C}(\phi)|$, $k_\phi = 5|\mathcal{C}(\phi)|$, and
  $\Delta(G_\phi)=4$.  Consequently, even on instances with maximum
  degree $4$, \tpe{} remains \cclass{NP}-hard and cannot be solved in
  time $2^{o(k)}n^{O(1)}$ or $2^{o(n+m)}$, unless ETH fails.
\end{theorem}

Theorem~\ref{thm:np-hardness} clearly refines Theorem~\ref{thm:eth-hardness-intro}, and its conclusion follows from the
reduction by an application of Proposition~\ref{prop:eth}.  Hence, we
are left with constructing the reduction, to which the rest of this
section is devoted.  Our approach is similar to the technique used by
Komusiewicz and Uhlmann to show the hardness of a similar problem,
\pname{Cluster Editing}~\cite{komusiewicz2012clusterediting}; However,
the gadgets are heavily modified to work for the \tpe{} problem.

Let $\phi$ be the input instance of \pname{3Sat}.  By standard
modifications of the formula we may assume that every clause contains
exactly three literals, all containing different variables, and that
every variable appears in at least two clauses.  For a variable $x\in
\mathcal{V}(\phi)$, let~$p_x>1$ be the number of occurrences of~$x$ in
the clauses of~$\phi$; Moreover, we order these occurrences
arbitrarily.  Observe that $\sum_{x\in \mathcal{V}(\phi)} p_x =
3|\mathcal{C}(\phi)|$.  Now, for every $x\in \mathcal{V}(\phi)$ we
create a \emph{variable gadget}, and for every $c\in \mathcal{C}(\phi)$ we
create a \emph{clause gadget}.


%
%
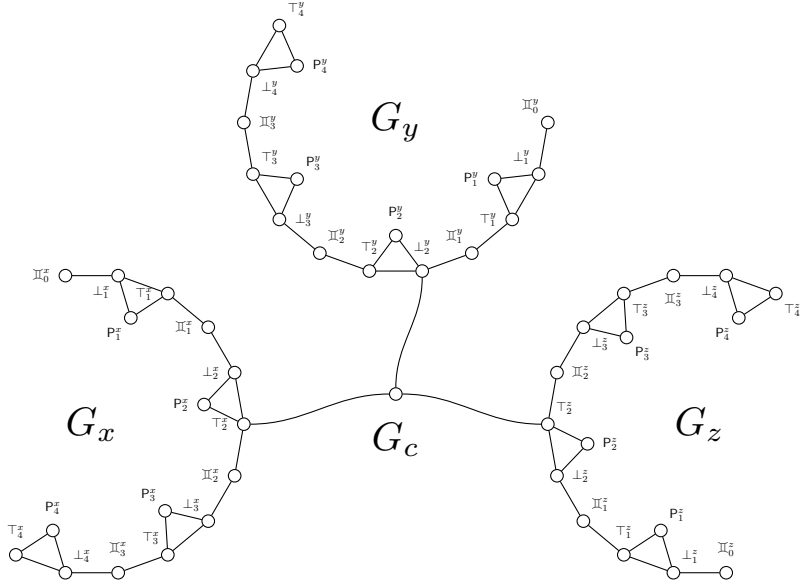
\begin{figure}[t]
  \centering
  \begin{tikzpicture}[every node/.style={circle, draw, fill=white, scale=.5},
    scale=2]
    
    \node (tx1) at ({cos(60)},{sin(60)}) [label=left:$\top_1^x$] {};
    \node (bx1) at ({cos(80)},{sin(80)}) [label=below left:$\bot_1^x$] {};
    \node (sx1) at ({cos(100)},{sin(100)}) [label=left:$\third_0^x$] {};
    \node (px1) at ({0.75*cos(70)},{0.75*sin(70)}) [label=below left:$\paw_1^x$] {};
    
    \node (tx2) at ({cos(0)},{sin(0)}) [label=left:$\top_2^x$] {};
    \node (bx2) at ({cos(20)},{sin(20)}) [label=left:$\bot_2^x$] {};
    \node (sx2) at ({cos(40)},{sin(40)}) [label=left:$\third_1^x$] {};
    \node (px2) at ({0.75*cos(10)},{0.75*sin(10)}) [label=left:$\paw_2^x$] {};
    
    \node (tx3) at ({cos(-60)},{sin(-60}) [label=above left:$\top_3^x$] {};
    \node (bx3) at ({cos(-40)},{sin(-40}) [label=above left:$\bot_3^x$] {};
    \node (sx3) at ({cos(-20)},{sin(-20}) [label=left:$\third_2^x$] {};
    \node (px3) at ({0.75*cos(-50)},{0.75*sin(-50)}) [label=above left:$\paw_3^x$] {};
    
    \node (tx4) at ({cos(-120)},{sin(-120)}) [label=above:$\top_4^x$] {};
    \node (bx4) at ({cos(-100)},{sin(-100)}) [label=above right:$\bot_4^x$] {};
    \node (sx4) at ({cos(-80)},{sin(-80)})   [label=above:$\third_3^x$] {};
    \node (px4) at ({0.75*cos(-110)},{0.75*sin(-110)}) [label=above:$\paw_4^x$] {};

    \node[draw=none,scale=3] (gx) at (0,0) {$G_x$};
    
    \draw (sx1) -- (bx1) -- (tx1) -- (sx2) -- (bx2) -- (tx2) -- (sx3)
    -- (bx3) -- (tx3) -- (sx4) -- (bx4) -- (tx4) ;
    
    \draw (tx1) -- (px1) -- (bx1);
    \draw (tx2) -- (px2) -- (bx2);
    \draw (tx3) -- (px3) -- (bx3);
    \draw (tx4) -- (px4) -- (bx4);

    \node (ty1) at ({2+cos(-40)},{2+sin(-40)})   [label=left:$\top_1^y$] {};
    \node (by1) at ({2+cos(-20)},{2+sin(-20)})   [label=above left:$\bot_1^y$] {};
    \node (sy1) at ({2+cos(0)},{2+sin(0)})       [label=above left:$\third_0^y$] {};
    \node (py1) at ({2+0.75*cos(-30)},{2+0.75*sin(-30)}) [label=left:$\paw_1^y$] {};
    
    \node (ty2) at ({2+cos(-100)},{2+sin(-100)}) [label=above:$\top_2^y$] {};
    \node (by2) at ({2+cos(-80)},{2+sin(-80)})   [label=above:$\bot_2^y$] {};
    \node (sy2) at ({2+cos(-60)},{2+sin(-60)})   [label=above left:$\third_1^y$] {};
    \node (py2) at ({2+0.75*cos(-90)},{2+0.75*sin(-90)}) [label=above:$\paw_2^y$] {};
    
    \node (ty3) at ({2+cos(-160)},{2+sin(-160})  [label=above right:$\top_3^y$] {};
    \node (by3) at ({2+cos(-140)},{2+sin(-140})  [label=right:$\bot_3^y$] {};
    \node (sy3) at ({2+cos(-120)},{2+sin(-120})  [label=above right:$\third_2^y$] {};
    \node (py3) at ({2+0.75*cos(-150)},{2+0.75*sin(-150)}) [label=above right:$\paw_3^y$] {};
    
    \node (ty4) at ({2+cos(-220)},{2+sin(-220)}) [label=above right:$\top_4^y$] {};
    \node (by4) at ({2+cos(-200)},{2+sin(-200)}) [label=below right:$\bot_4^y$] {};
    \node (sy4) at ({2+cos(-180)},{2+sin(-180)}) [label=right:$\third_3^y$] {};
    \node (py4) at ({2+0.75*cos(-210)},{2+0.75*sin(-210)}) [label=right:$\paw_4^y$] {};

    \node[draw=none,scale=3] (gy) at (2,2) {$G_y$};    

    \draw (sy1) -- (by1) -- (ty1) -- (sy2) -- (by2) -- (ty2) -- (sy3)
    -- (by3) -- (ty3) -- (sy4) -- (by4) -- (ty4) ;

    \draw (ty1) -- (py1) -- (by1);
    \draw (ty2) -- (py2) -- (by2);
    \draw (ty3) -- (py3) -- (by3);
    \draw (ty4) -- (py4) -- (by4);

    \node (tz1) at ({4+cos(240)},{sin(240)}) [label=above:$\top_1^z$] {};
    \node (bz1) at ({4+cos(260)},{sin(260)}) [label=above right:$\bot_1^z$] {};
    \node (sz1) at ({4+cos(280)},{sin(280)}) [label=above:$\third_0^z$] {};
    \node (pz1) at ({4+0.75*cos(250)},{0.75*sin(250)}) [label=above right:$\paw_1^z$] {};
    
    \node (tz2) at ({4+cos(180)},{sin(180)}) [label=above right:$\top_2^z$] {};
    \node (bz2) at ({4+cos(200)},{sin(200)}) [label=right:$\bot_2^z$] {};
    \node (sz2) at ({4+cos(220)},{sin(220)}) [label=above right:$\third_1^z$] {};
    \node (pz2) at ({4+0.75*cos(190)},{0.75*sin(190)})    [label=right:$\paw_2^z$] {};
    
    \node (tz3) at ({4+cos(120)},{sin(120}) [label=below right:$\top_3^z$] {};
    \node (bz3) at ({4+cos(140)},{sin(140}) [label=below right:$\bot_3^z$] {};
    \node (sz3) at ({4+cos(160)},{sin(160}) [label=right:$\third_2^z$] {};
    \node (pz3) at ({4+0.75*cos(130)},{0.75*sin(130)}) [label=below right:$\paw_3^z$] {};
    
    \node (tz4) at ({4+cos(60)},{sin(60)}) [label=below right:$\top_4^z$] {};
    \node (bz4) at ({4+cos(80)},{sin(80)}) [label=below left:$\bot_4^z$] {};
    \node (sz4) at ({4+cos(100)},{sin(100)})   [label=below:$\third_3^z$] {};
    \node (pz4) at ({4+0.75*cos(70)},{0.75*sin(70)}) [label=below left:$\paw_4^z$] {};

    \node[draw=none,scale=3] (gz) at (4,0) {$G_z$};    

    \draw (sz1) -- (bz1) -- (tz1) -- (sz2) -- (bz2) -- (tz2) -- (sz3)
    -- (bz3) -- (tz3) -- (sz4) -- (bz4) -- (tz4) ;

    \draw (tz1) -- (pz1) -- (bz1);
    \draw (tz2) -- (pz2) -- (bz2);
    \draw (tz3) -- (pz3) -- (bz3);
    \draw (tz4) -- (pz4) -- (bz4);
    
    \node[draw=none,scale=3] (gc) at (2,-.1) {$G_c$};
    \node (c) at (2,.2) {};

    \draw (c) to[looseness=1, out=180, in=0] (tx2);
    \draw (c) to[looseness=1, out=90, in=-90] (by2);
    \draw (c) to[looseness=1, out=0, in=180] (tz2);
    
  \end{tikzpicture}
  \caption{Gadget $c = x \lor \neg y \lor z$.  The clause $c$ is now the second
    clause all variables $x$, $y$, and $z$ appear in, and $x$ and $z$
    appears positively whereas $y$ appears negatively.}
  \label{fig:clause-gadget}
\end{figure}

\begin{figure}[t]
  \centering
  \begin{tikzpicture}[every node/.style={circle, draw, fill=white, scale=.5},
    scale=2]

    \node (tx1) at ({cos(60)},{sin(60)}) [label=left:$\top_1^x$] {};
    \node (bx1) at ({cos(80)},{sin(80)}) [label=below left:$\bot_1^x$] {};
    \node (sx1) at ({cos(100)},{sin(100)}) [label=left:$\third_0^x$] {};
    \node (px1) at ({0.75*cos(70)},{0.75*sin(70)}) [label=below left:$\paw_1^x$] {};
    
    \node (tx2) at ({cos(0)},{sin(0)}) [label=left:$\top_2^x$] {};
    \node (bx2) at ({cos(20)},{sin(20)}) [label=left:$\bot_2^x$] {};
    \node (sx2) at ({cos(40)},{sin(40)}) [label=left:$\third_1^x$] {};
    \node (px2) at ({0.75*cos(10)},{0.75*sin(10)}) [label=left:$\paw_2^x$] {};
    
    \node (tx3) at ({cos(-60)},{sin(-60}) [label=above left:$\top_3^x$] {};
    \node (bx3) at ({cos(-40)},{sin(-40}) [label=above left:$\bot_3^x$] {};
    \node (sx3) at ({cos(-20)},{sin(-20}) [label=left:$\third_2^x$] {};
    \node (px3) at ({0.75*cos(-50)},{0.75*sin(-50)}) [label=above left:$\paw_3^x$] {};
    
    \node (tx4) at ({cos(-120)},{sin(-120)}) [label=above:$\top_4^x$] {};
    \node (bx4) at ({cos(-100)},{sin(-100)}) [label=above right:$\bot_4^x$] {};
    \node (sx4) at ({cos(-80)},{sin(-80)})   [label=above:$\third_3^x$] {};
    \node (px4) at ({0.75*cos(-110)},{0.75*sin(-110)}) [label=above:$\paw_4^x$] {};

    \node[draw=none,scale=3] (gx) at (0,0) {$G_x$};
        
    \draw (bx1) -- (tx1) -- (sx2);
    \draw (bx2) -- (tx2) -- (sx3);
    \draw (bx3) -- (tx3) -- (sx4);
    \draw (bx4) -- (tx4) ;
    
    \draw (tx1) -- (px1) -- (bx1);
    \draw (tx2) -- (px2) -- (bx2);
    \draw (tx3) -- (px3) -- (bx3);
    \draw (tx4) -- (px4) -- (bx4);

    \node (ty1) at ({2+cos(-40)},{2+sin(-40)})   [label=left:$\top_1^y$] {};
    \node (by1) at ({2+cos(-20)},{2+sin(-20)})   [label=above left:$\bot_1^y$] {};
    \node (sy1) at ({2+cos(0)},{2+sin(0)})       [label=above left:$\third_0^y$] {};
    \node (py1) at ({2+0.75*cos(-30)},{2+0.75*sin(-30)}) [label=left:$\paw_1^y$] {};
    
    \node (ty2) at ({2+cos(-100)},{2+sin(-100)}) [label=above:$\top_2^y$] {};
    \node (by2) at ({2+cos(-80)},{2+sin(-80)})   [label=above:$\bot_2^y$] {};
    \node (sy2) at ({2+cos(-60)},{2+sin(-60)})   [label=above left:$\third_1^y$] {};
    \node (py2) at ({2+0.75*cos(-90)},{2+0.75*sin(-90)}) [label=above:$\paw_2^y$] {};
    
    \node (ty3) at ({2+cos(-160)},{2+sin(-160})  [label=above right:$\top_3^y$] {};
    \node (by3) at ({2+cos(-140)},{2+sin(-140})  [label=right:$\bot_3^y$] {};
    \node (sy3) at ({2+cos(-120)},{2+sin(-120})  [label=above right:$\third_2^y$] {};
    \node (py3) at ({2+0.75*cos(-150)},{2+0.75*sin(-150)}) [label=above right:$\paw_3^y$] {};
    
    \node (ty4) at ({2+cos(-220)},{2+sin(-220)}) [label=above right:$\top_4^y$] {};
    \node (by4) at ({2+cos(-200)},{2+sin(-200)}) [label=below right:$\bot_4^y$] {};
    \node (sy4) at ({2+cos(-180)},{2+sin(-180)}) [label=right:$\third_3^y$] {};
    \node (py4) at ({2+0.75*cos(-210)},{2+0.75*sin(-210)}) [label=right:$\paw_4^y$] {};

    \node[draw=none,scale=3] (gy) at (2,2) {$G_y$};

    \draw (by1) -- (ty1) -- (sy2);
    \draw (by2) -- (ty2) -- (sy3);
    \draw (by3) -- (ty3) -- (sy4);
    \draw (by4) -- (ty4) ;
    
    \draw (ty1) -- (py1) -- (by1);
    \draw (ty2) -- (py2) -- (by2);
    \draw (ty3) -- (py3) -- (by3);
    \draw (ty4) -- (py4) -- (by4);

    \node (tz1) at ({4+cos(240)},{sin(240)}) [label=above:$\top_1^z$] {};
    \node (bz1) at ({4+cos(260)},{sin(260)}) [label=above right:$\bot_1^z$] {};
    \node (sz1) at ({4+cos(280)},{sin(280)}) [label=above:$\third_0^z$] {};
    \node (pz1) at ({4+0.75*cos(250)},{0.75*sin(250)}) [label=above right:$\paw_1^z$] {};
    
    \node (tz2) at ({4+cos(180)},{sin(180)}) [label=above right:$\top_2^z$] {};
    \node (bz2) at ({4+cos(200)},{sin(200)}) [label=right:$\bot_2^z$] {};
    \node (sz2) at ({4+cos(220)},{sin(220)}) [label=above right:$\third_1^z$] {};
    \node (pz2) at ({4+0.75*cos(190)},{0.75*sin(190)})    [label=right:$\paw_2^z$] {};
    
    \node (tz3) at ({4+cos(120)},{sin(120}) [label=below right:$\top_3^z$] {};
    \node (bz3) at ({4+cos(140)},{sin(140}) [label=below right:$\bot_3^z$] {};
    \node (sz3) at ({4+cos(160)},{sin(160}) [label=right:$\third_2^z$] {};
    \node (pz3) at ({4+0.75*cos(130)},{0.75*sin(130)}) [label=below right:$\paw_3^z$] {};
    
    \node (tz4) at ({4+cos(60)},{sin(60)}) [label=below right:$\top_4^z$] {};
    \node (bz4) at ({4+cos(80)},{sin(80)}) [label=below left:$\bot_4^z$] {};
    \node (sz4) at ({4+cos(100)},{sin(100)})   [label=below:$\third_3^z$] {};
    \node (pz4) at ({4+0.75*cos(70)},{0.75*sin(70)}) [label=below left:$\paw_4^z$] {};

    \node[draw=none,scale=3] (gz) at (4,0) {$G_z$};    
    
    \draw (sz1) -- (bz1) -- (tz1);
    \draw (sz2) -- (bz2) -- (tz2);
    \draw (sz3) -- (bz3) -- (tz3);
    \draw (sz4) -- (bz4) -- (tz4);

    \draw (tz1) -- (pz1) -- (bz1);
    \draw (tz2) -- (pz2) -- (bz2);
    \draw (tz3) -- (pz3) -- (bz3);
    \draw (tz4) -- (pz4) -- (bz4);

    \node[draw=none,scale=3] (gc) at (2,-.1) {$G_c$};
    \node (c) at (2,.2) {};
    
    \draw (c) to[looseness=1, out=180, in=0] (tx2);
    
  \end{tikzpicture}
  \caption{Edited gadget of $c = x \lor \neg y \lor z$ where $\alpha(x) = \top$, $\alpha(y) =
    \top$ and $\alpha(z) = \bot$ and $x$ has been chosen (no choice)
    to satisfy $c$.  Notice the formation of \emph{paws}, except the
    one incident to $c$ which induces a \emph{cricket}.}
  \label{fig:clause-gadget-edited}
\end{figure}
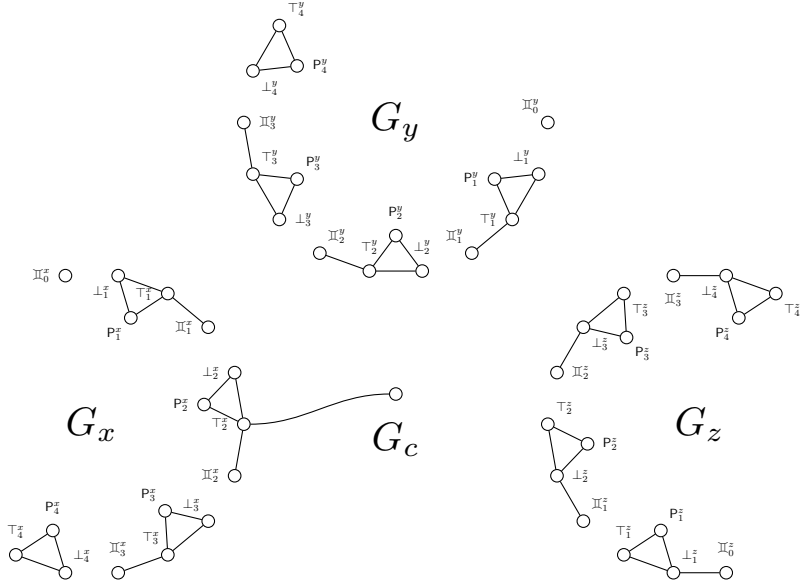

\paragraph{Variable gadgets.}  For $x \in \mathcal{V}(\phi)$, construct a
graph $G_x$ isomorphic to $C_{3p_x}$, a cycle on $3p_x$ vertices.  The
vertices of $G_x$ are labeled $\bot^x_i,\top^x_i,\third^x_i$ for $i
\in [0, p_x - 1]$, in the order of their appearance on the cycle.  We
then add a vertex $\paw^x_i$ adjacent to $\top^x_i$ and $\bot^x_i$,
for each $i\in [0,p_x-1]$, see Figure~\ref{fig:clause-gadget}.
Formally, the vertices $\paw^x_i$ do not belong to $G_x$, but they
will be used to wire variable gadgets with clause gadgets.  This
concludes the construction of the variable gadget, and it should be
clear that the number of created vertices and edges is bounded
linearly in $p_x$; More precisely, we created $4p_x$ vertices and
$5p_x$ edges.

For the sake of later argumentation, we now define the deletion set
$F^\alpha_x$ for $G_x$.  If, in an assignment of variables $\alpha:
\mathcal{V}(\phi) \to \{\top,\bot\}$, we have $\alpha(x) = \top$, then
we let $F^\alpha_x$ be the set consisting of every edge of the form
$\third^x_i\bot^x_{i+1 \bmod p_x}$ for $i\in [0, p_x - 1]$.  If, on
the other hand, $\alpha(x) = \bot$, we define the deletion set
$F^\alpha_x$ to be the set comprising the edges $\top^x_i\third^x_i$
for $i\in [0, p_x - 1]$, see Figure~\ref{fig:clause-gadget-edited}.
We will later show that these are the only relevant editing sets of
size at most $p_x$ for $G_x$.

\paragraph{Clause gadget.} The clause gadgets are very simple.  A
clause gadget consists simply of one vertex, i.e., for a clause $c \in
\mathcal{C}(\phi)$ construct the vertex~$v_c$.  This vertex will be
connected to~$G_x$,~$G_y$ and~$G_z$, for~$x$,~$y$, and~$z$ being the
variables appearing in~$c$, in appropriate places, depending on
whether the variable occurs positively or negatively in~$c$.  More
precisely, if~$c$ is the~$i$th clause~$x$ appears in, then we
make~$v_c$ adjacent to~$\top^x_i$ provided that~$x$ appears positively
in~$c$, and to~$\bot^x_i$ provided that~$x$ appears negatively in~$c$.
This concludes the construction of a clause gadget.  As every clause
gadget contains one vertex and three edges, the construction of all
the clause gadgets creates $|\mathcal{C}(\phi)|$ vertices and
$3|\mathcal{C}(\phi)|$ edges.

The deletion set for a clause gadget will be as follows.  Let $\alpha:
\mathcal{V}(\phi) \to \{\top,\bot\}$, be an assignment of the
variables that satisfies all the clauses.  Suppose $c = \ell_x \lor
\ell_y \lor \ell_z$, where the literals~$\ell_x$,~$\ell_y$,
and~$\ell_z$ contain variables~$x$,~$y$, and~$z$, respectively.  Pick
any literal satisfying~$c$, say~$\ell_x$, and delete the two other
edges in the connection, i.e., the two edges connecting~$v_c$ with
vertices of~$G_y$ and~$G_z$.  Thus~$v_c$ remains a vertex of
degree~$1$, adjacent to a vertex of~$G_x$.

\bigskip

Let $G_\phi$ be the constructed graph.  We set the budget for edits to
\begin{eqnarray*}
k_\phi =& \sum_{x\in \mathcal{V}(\phi)} p_x + 2|\mathcal{C}(\phi)| =5 |\mathcal{C}(\phi)|.
\end{eqnarray*}
Observe also that 
\begin{eqnarray*}
|V(G_\phi)|=& \sum_{x\in \mathcal{V}(\phi)} 4p_x + |\mathcal{C}(\phi)|=13|\mathcal{C}(\phi)|,\\
|E(G_\phi)|=& \sum_{x\in \mathcal{V}(\phi)} 5p_x + 3|\mathcal{C}(\phi)|=18|\mathcal{C}(\phi)|,
\end{eqnarray*}
and that $\Delta(G_\phi)=4$.  Thus, all the technical properties stated in
Theorem~\ref{thm:np-hardness} are satisfied, and we are left with
proving that $(G_\phi,k_\phi)$ is a yes-instance of \tpe{} if and only if
$\phi$ is satisfiable.

Before we state the main lemma, we give two auxiliary observations
that settle the tightness of the budget:
\begin{claim}
  \label{claim:var-gadget-tight}
  Suppose that a graph $H$ is a cycle on $3p$ vertices for some $p>1$,
  and suppose~$F$ is an editing set for~$H$.  Then $|F|\geq p$.
  Moreover, if~$|F|=p$ then~$F$ consists of deletions of every third
  edge of the cycle.
\end{claim}
\begin{claim}
  \label{claim:clause-gadget-tight}
  Suppose a graph $H$ is a \emph{subdivided claw}, i.e., the star
  $K_{1,3}$ with every leg subdivided once (see
  Figure~\ref{fig:subdivided-claw}).  Furthermore, suppose that~$F$ is
  an editing set for~$H$.  Then~$|F|\geq 2$.  Moreover, if~$|F|=2$
  then~$F$ consists of deletions of two edges incident to the center
  of the subdivided claw (see
  Figure~\ref{fig:subdivided-claw-edited}).
\end{claim}


We will prove the two claims in order now.  The astute reader should
already see that this implies the tightness of the budget: every
editing set needs to include exactly~$p_x$ edges of every variable
gadget~$G_x$ (by Claim~\ref{claim:var-gadget-tight}), and exactly two
edges incident to every vertex~$v_c$ (by
Claim~\ref{claim:clause-gadget-tight}).  The additional
vertices~$\paw^x_i$ will form the degree-1 vertices of subdivided
claws created by clause gadgets, and all the subgraphs in question
pairwise share at most single vertices, which means that any edit can
influence at most one of them.  This statement is made formal in the
proof of Lemma~\ref{lem:3sat-tpe}.

%
%
%
%
\begin{proof}[Proof of Claim~\ref{claim:var-gadget-tight}]
  Let $v_0,v_1,\ldots,v_{3p-1}$ be the vertices of~$H$, in their order
  of appearance on the cycle.  For $i=0,1,\ldots,p-1$, let
  $A_i=\{v_{3i},v_{3i+1},v_{3i+2},v_{3i+3}\}$; Here and in the sequel,
  the indices behave cyclically in a natural manner.  Observe that
  each $A_i$ induces a $P_4$ in $H$, hence $F\cap \binom{A_i}{2}\neq
  \emptyset$.  However, the sets $\binom{A_i}{2}$ are pairwise
  disjoint for $i=0,1,\ldots,p-1$, from which it follows that $|F|\geq
  p$.

  Suppose now that $|F|=p$.  Hence $|F\cap \binom{A_i}{2}|=1$ for each
  $i\in [0,p-1]$, and there are no edits outside the sets $\binom{A_i}{2}$.
  There are five possible ways for an $A_i$ of how $F\cap
  \binom{A_i}{2}$ can look like: It is either a deletion of the edge
  $v_{3i}v_{3i+1}$, $v_{3i+1}v_{3i+2}$, or $v_{3i+2}v_{3i+3}$
  (henceforth referred to as types~$D^-$,~$D^0$, and~$D^+$,
  respectively), or an addition of the edge $v_{3i}v_{3i+2}$ or
  $v_{3i+1}v_{3i+3}$ (henceforth called types~$C^-$ and~$C^+$,
  respectively)---the sixth possibility, which has been left out,
  creates an induced~$C_4$.  Observe now that if some~$A_i$ has
  type~$D^-$, then~$A_{i+1}$ also has type~$D^-$, or otherwise a~$P_4$
  $v_{3i+1}-v_{3i+2}-v_{3i+3}-v_{3i+4}$ would remain in the graph.
  Similarly, if~$A_i$ has type~$D^+$ then~$A_{i-1}$ also has
  type~$D^+$.  Hence, if type~$D^+$ or~$D^-$ appears for any~$A_i$,
  then all the~$A_i$s have the same type.  Observe now that if
  some~$A_i$ had type~$C^-$ and~$C^+$, then~$A_{i-1}$ would have to
  have type~$D^+$ and~$A_{i+1}$ would have to have type~$D^-$ or
  otherwise an unresolved~$P_4$ would appear; This is a contradiction
  with the previous observations, since types~$D^-$ and~$D^+$ cannot
  appear simultaneously.  Hence, we are left with only three
  possibilities: all the~$A_i$s have type~$D^-$, or all have
  type~$D^0$, or all have type~$D^+$.  \cqed\end{proof}

\begin{proof}[Proof of Claim~\ref{claim:clause-gadget-tight}]
  Denote the vertices of $H$ as in Figure~\ref{fig:subdivided-claw}.
  Consider the following three $P_4$s in $H$:
  \begin{itemize}
  \item $a_2-a_1-v-c_1$,
  \item $b_2-b_1-v-a_1$, and
  \item $c_2-c_1-v-b_1$.
  \end{itemize}

  \begin{figure}[htp]
    \centering
    \begin{subfigure}[t]{0.45\textwidth}
      \centering
      \begin{tikzpicture}[every node/.style={circle, draw, scale=.7},
        minimum size=2em, scale=1, inner sep=0]
        
        \def\so{0.707106} 
        \def\st{1.4142135} 
        
        \node (v)  at (0,0) {$v$};
        \node (a1) at (0,1) {$a_1$};
        \node (a2) at (0,2) {$a_2$};
        \node (b1) at (-\so,-\so) {$b_1$};
        \node (b2) at (-\st,-\st) {$b_2$};
        \node (c1) at (\so,-\so) {$c_1$};
        \node (c2) at (\st,-\st) {$c_2$};
        
        \draw (a2) -- (a1) -- (v) -- (b1) -- (b2);
        \draw (v) -- (c1) -- (c2);
      \end{tikzpicture}
      \caption{A subdivided claw.}
      \label{fig:subdivided-claw}
    \end{subfigure}
    \hspace{.03\textwidth}
    \begin{subfigure}[t]{0.45\textwidth}
      \centering
      \begin{tikzpicture}[every node/.style={circle, draw, scale=.7},
        minimum size=2em, scale=1, inner sep=0]
        
        \def\so{0.707106} 
        \def\st{1.4142135} 
        
        \node (v)  at (0,0) {$v$};
        \node (a1) at (0,1) {$a_1$};
        \node (a2) at (0,2) {$a_2$};
        \node (b1) at (-\so,-\so) {$b_1$};
        \node (b2) at (-\st,-\st) {$b_2$};
        \node (c1) at (\so,-\so) {$c_1$};
        \node (c2) at (\st,-\st) {$c_2$};
        
        \draw (a2) -- (a1) -- (v);
        \draw (b1) -- (b2);
        \draw (c1) -- (c2);
      \end{tikzpicture}
      \caption{An optimally edited subdivided claw.}
      \label{fig:subdivided-claw-edited}
    \end{subfigure}
  \end{figure}

  Observe that any edge addition in $H$ can destroy at most one of
  these~$P_4$s, and a deletion of any of edges $a_1a_2$, $b_1b_2$, or
  $c_1c_2$ also can destroy at most one of these~$P_4$s.  Moreover, a
  deletion of any of the edges incident to the center~$v$ destroys
  only two of them.  We infer that~$|F|\geq 2$ since no single edit
  can destroy all three considered~$P_4$s, and moreover if~$|F|=2$,
  then~$F$ contains at least one deletion of an edge incident to~$v$,
  say~$va_1$.  After deleting this edge we are left with a~$P_5$
  $b_2-b_1-v-c_1-c_2$, and it can be readily checked that the only way
  to edit it to a trivially perfect graph using only one edit is to
  delete~$vb_1$ or~$vc_1$.  Thus, any editing set~$F$ with~$|F|=2$ in
  fact consists of deletions of two edges incident
  to~$v$.\cqed\end{proof}

\begin{lemma}
  \label{lem:3sat-tpe}
  The input \pname{3Sat} instance $\phi$ is satisfiable if and only if $(G_\phi,k_\phi)$
  is a yes-instance of \tpe.
\end{lemma}
\begin{proof}
  Suppose $\phi$ is satisfiable and let $\alpha: \mathcal{V}(\phi) \to \{\top,\bot\}$
  be a satisfying assignment.  Define editing set
  $F^\alpha=\bigcup_{x\in \mathcal{V}(\phi)} F^\alpha_x\cup
  \bigcup_{c\in \mathcal{C}(\phi)} F^\alpha_c$; Note that~$F$ consists
  of deletions only.  Then we have that~$|F^\alpha|=k_\phi$ and it can
  be easily seen that~$G\triangle F$ is a disjoint union of components
  of constant size, each being a paw or a cricket (see
  Figure~\ref{fig:three-graphs}).  Both these graphs are trivially
  perfect, so a disjoint union of any number of their copies is also a
  trivially perfect graph.  Thus~$F^\alpha$ is a solution to the
  instance~$(G_\phi,k_\phi)$.
  
  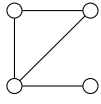
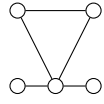
\begin{figure}[htp]
    \centering
    \begin{subfigure}[t]{0.3\textwidth}
      \centering
      \begin{tikzpicture}[every node/.style={circle, draw, scale=.6},
        scale=1]
        \node (1) at (0,0) {};
        \node (2) at (1,0) {};
        \node (3) at (0,1) {};
        \node (4) at (1,1) {};
        
        \draw (1) -- (2);
        \draw (1) -- (3);
        \draw (4) -- (3);
        \draw (1) -- (4);
      \end{tikzpicture}
      \caption{Paw}
    \end{subfigure}
    \hspace{.03\textwidth}
    \begin{subfigure}[t]{0.3\textwidth}
      \centering
      \begin{tikzpicture}[every node/.style={circle, draw, scale=.6},
        scale=1]
        \node (1) at (.5,0) {};
        \node (2) at (0,1) {};
        \node (3) at (1,1) {};
        \node (4) at (1,0) {};
        \node (5) at (0,0) {};
        
        \draw (1) -- (2) -- (3) -- (1);
        \draw (1) -- (4);
        \draw (1) -- (5);
      \end{tikzpicture}
      \caption{Cricket}
    \end{subfigure}
    \caption{Shapes of components of $G$ after editing deletion sets
      $F^\alpha_x$ and $F^\alpha_c$ for $\alpha$ being a satisfying
      assignment.  Both of them are trivially perfect, so a disjoint
      union of any number of their copies is also trivially
      perfect.} \label{fig:three-graphs}
  \end{figure}
  
  For the other direction, let $F\subseteq \binom{V(G_\phi)}{2}$ be an editing
  set such that $G_\phi \triangle F$ is trivially perfect, and $|F|
  \leq k_\phi$.  For every $x\in \mathcal{V}(\phi)$ consider the
  subgraph $G_x$.  For every $c\in \mathcal{C}(\phi)$ consider the
  subgraph $G_c$ induced in $G$ by
  \begin{itemize}
  \item vertex $v_c$;
  \item the three neighbors of $v_c$, say $\Box^x_{i_x}$, $\Box^y_{i_y}$, and
    $\Box^z_{i_z}$, where $x,y,z$ are variables appearing in $c$ and
    each symbol $\Box$ is replaced by $\bot$ or $\top$ depending whether
    the variable's occurrence is positive or negative; and
  \item vertices $\paw^x_{i_x}$, $\paw^y_{i_y}$, and $\paw^z_{i_z}$.
  \end{itemize}
  Observe that each $G_x$ is isomorphic to a cycle on $3p_x$ vertices
  and each $G_c$ is isomorphic to a subdivided claw.  Moreover, all
  these subgraphs pairwise share at most one vertex, which means that
  sets $\binom{V(G_x)}{2}$ for $x\in \mathcal{V}(\phi)$ and
  $\binom{V(G_c)}{2}$ for $c\in \mathcal{C}(\phi)$ are pairwise
  disjoint.  By Claim~\ref{claim:var-gadget-tight} we infer that
  $|F\cap \binom{V(G_x)}{2}|\geq p_x$ for each $x\in
  \mathcal{V}(\phi)$, and by Claim~\ref{claim:clause-gadget-tight} we
  infer that $|F\cap \binom{V(G_c)}{2}|\geq 2$ for each $c\in
  \mathcal{C}(\phi)$.  Thus
  $$|F|\geq \sum_{x\in \mathcal{V}(\phi)} p_x+2|\mathcal{C}(\phi)|=k_\phi.$$
  Hence, in fact $|F|=k_\phi$ and all the used inequalities are in fact
  equalities: $|F\cap \binom{V(G_x)}{2}|=p_x$ for each $x\in
  \mathcal{V}(\phi)$ and $|F\cap \binom{V(G_c)}{2}|=2$ for each $c\in
  \mathcal{C}(\phi)$.  Using Claims~\ref{claim:var-gadget-tight}
  and~\ref{claim:clause-gadget-tight} again, we infer that $F$ has the
  following form: it consists of deletions only, from every cycle
  $G_x$ it deletes every third edge, and for every vertex $v_c$ it
  deletes two out of three edges incident to it.  In particular, no
  edit is incident to any of the vertices $\paw^x_i$ for $x\in
  \mathcal{V}(\phi)$ and $i\in [0,p_x-1]$.
  
  Consider now the cycle $G_x$; We already know that the solution
  deletes either all the edges $\bot^x_i\top^x_{i}$ for $i\in
  [0,p_x-1]$, or all the edges $\top^x_i\third^x_{i}$ for $i\in
  [0,p_x-1]$, or all the edges $\third^x_i\bot^x_{i+1\bmod p_x}$ for
  $i\in [0,p_x-1]$.  Observe that the first case cannot happen, since
  then we would have an induced $P_4$
  $\bot^x_i-\paw^x_i-\top^x_i-\third^x_i$ remaining in the graph ---
  no other edit can destroy it.  Hence, one of the latter two cases
  happen.  Construct an assignment $\alpha: \mathcal{V}(\phi) \to
  \{\top,\bot\}$ by, for each $x\in \mathcal{V}(\phi)$, putting
  $\alpha(x)=\bot$ if all the edges $\top^x_i\third^x_{i}$ are
  included in $F$, and $\alpha(x)=\top$ if all the edges
  $\third^x_i\bot^x_{i+1\bmod p_x}$ are included in $F$.  We now claim
  that $\alpha$ satisfies $\phi$.
  
  
  For the sake of contradiction, suppose that a clause $c=\ell_x\vee
  \ell_y\vee \ell_z$ is not satisfied by $\alpha$.  Let $e$ be the
  edge incident to $v_c$ which has not been removed and suppose
  without loss of generality that this edge connects $v_c$ with $G_x$.
  Suppose further that $\ell_x=x$, i.e., $x$ appears positively in
  $c$, so $e=v_c\top^x_i$ for some $i\in [0,p_x-1]$.  Since $x$ does
  not satisfy $c$, $\alpha(x) = \bot$ and both edges $\third^x_{i-1
    \bmod p_x}\bot^x_i$ and $\bot^x_i\top^x_i$ are not deleted in $F$
  --- the deleted edge is $\top^x_i\third^x_i$.  But then we have the
  following induced $P_4$: $v_c-\top^x_i-\bot^x_i-\third^x_{i-1 \mod
    p_x}$, which contradicts the assumption that $G_\phi \triangle F$
  is trivially perfect.  The case when $\ell_x=\neg x$, i.e., $x$
  appears negatively in $c$, is symmetric.
  
  Hence $\alpha$ is indeed a satisfying assignment for $\phi$ and we are done.
\end{proof}

Lemma~\ref{lem:3sat-tpe} guarantees that the reduction is correct, and
hence Theorem~\ref{thm:np-hardness} follows by a straightforward
application of Proposition~\ref{prop:eth}.  We can also observe that
this reduction works immediately for \tpd{} as well since every
optimal edit set consisted purely of deletions (see
Claims~\ref{claim:var-gadget-tight}
and~\ref{claim:clause-gadget-tight}), however this result is
known~\cite{drange2014exploring}.

\section{Conclusion}
\label{sec:conclusion}

In this paper we gave the first polynomial kernels for \tpe{} and
\tpd, which answers an open problem by Nastos and
Gao~\cite{nastos2013familial}, and Liu, Wang, and
Guo~\cite{liu2014overview}.  We also proved that assuming ETH, \tpe{}
does not have a subexponential parameterized algorithm.  Together with
the earlier results~\cite{drange2014exploring,guo2007problem}, we thus
obtain a complete picture of the existence of polynomial kernels and
subexponential parameterized algorithms for edge modification problems
related to trivially perfect graphs; see
Figure~\ref{fig:tab:tp-complexity} for an overview.  In particular,
the fact that all three problems \tpe, \tpc, and \tpd{} admit
polynomial kernels, stands in an interesting contrast with the results
of Cai and Cai~\cite{cai2013incompressibility}, who showed that this
is not the case for any of \pname{$C_4$-Free Editing},
\pname{$C_4$-Free Completion} and \pname{$C_4$-Free Deletion}.

The main contribution of the paper is the proof that \TPE{} admits a
polynomial kernel with $O(k^7)$ vertices.  We apply the existing
technique of constructing a \emph{vertex modulator}, but with a new
twist: The fact that we are solving an edge modification problem
enables us also to argue about the adjacency structure between the
modulator and the rest of the graph, which is helpful in understanding
the structure of the instance.  We believe that this new insight can
be applied to other edge modification problems as well.

Finally, we showed that \TPE{}, in addition to being \NP-complete, is
not solvable in subexponential parameterized time unless the
Exponential Time Hypothesis fails.  The same result was known for
\tpd, but contrasts the previous result that the completion variant
\emph{does admit} a subexponential parameterized
algorithm~\cite{drange2014exploring}.

\begin{figure}[ht]
  \centering
  \begin{tabular}{l l l}
    Problem & Polynomial kernel & Subexp.\ par.\ algorithm\\
    \hline\\
    \tpc    & Yes~\cite{guo2007problem} & Yes~\cite{drange2014exploring}\\
    \tpd    & Yes & No~\cite{drange2014exploring}\\
    \tpe    & Yes & No\\
  \end{tabular}
  \caption{Graph modification problems related to trivially perfect
    graphs}
  \label{fig:tab:tp-complexity}
\end{figure}

\bigskip

Let us conclude by stating some open questions.  In this paper, we
focused purely on constructing a polynomial kernel for \tpe{} and
related problems, and in multiple places we traded possible savings in
the overall kernel size for simpler arguments in the analysis.  We
expect that a tighter analysis of our approach might yield kernels
with $O(k^6)$ or even $O(k^5)$ vertices, but we think that the really
challenging question is to match the size of the cubic kernel for
\tpc{} of Guo~\cite{guo2007problem}.

Generally, we find the vertex modulator technique very well-suited for
tackling kernelization of edge modification problems, since it is at
the same time versatile, and exposes well the structure of a large
graph that is close in the edit distance to some graph class.  We have
high hopes that this generic approach will find applications in other
edge modification problems as well, both in improving the sizes of
existing kernels and in finding new positive results about the
existence of polynomial kernels.  For concrete questions where the
technique might be applicable, we propose the following:
\begin{itemize}
\item Is it possible to improve the $O(k^3)$ vertex kernels for
  \pname{Cograph Editing} and \pname{Cograph Completion} of Guillemot
  et al.~\cite{guillemot2013onthenon}?
\item Is it possible to improve the $O(k^4)$ vertex kernel for the
  \pname{Split Deletion} problem of Guo~\cite{guo2007problem}?
\item Do the \pname{Claw-Free Edge Deletion} or \pname{Line Graph Edge
    Deletion} problems admit polynomial kernels? Here, the task is to
  remove at most~$k$ edges to obtain a graph that is \emph{claw-free},
  i.e., does not contain $K_{1,3}$ as an induced subgraph,
  respectively is a line graph.
\end{itemize}

\end{document}